\documentclass[reqno]{amsart}
\usepackage[all]{xy}
\usepackage{amsmath,amsfonts,amssymb,amsthm,epsfig,amscd,comment,latexsym,psfrag}
\usepackage{nicefrac,xspace,tikz}
\usetikzlibrary{arrows}

\newtheorem{prop}{Proposition}
\newtheorem{theorem}{Theorem}
\newtheorem{cor}{Corollary}

\newtheorem{lemma}{Lemma}

\theoremstyle{definition}
\newtheorem{Remark}{Remark}

\def\tr{\mbox{tr}}
\def\fP{\mathsf{P}}
\def\fQ{\mathsf{Q}}

\def\la{\langle}
\def\ra{\rangle}
\def\lba{[}
\def\rba{]}

\def\H{\mathcal{H}}

\def\Z{\mathbb Z}
\def\Zt{\Z[t,t^{-1}]}

\def\C{\mathbb C}
\def\Ct{\C[t,t^{-1}]}
\def\Ccrs{\C^{\times}}
\def\u1{\mathrm{U}(1)}
\def\l{\lbrace}

\def\lra{\longrightarrow}

\def\Sym{\rm Sym}
\def\Vv{V^{\vee}}
\def\gmod{\rm\bf -gmod}

\def\m{\mathbf m}
\def\n{\mathbf n}
\def\l{\mathbf l}
\def\id{\mathrm{id}}
\def\End{\mathrm{End}}
\def\adj{\mathrm{adj}}
\def\HHhom{\Hom_{'\H'}}

\def\lra{\longrightarrow}
\def\dmod{\rm\bf -mod}

\newcommand{\Hom}{{\rm Hom}}

\psfrag{Qpm}{$Q_{+-}$} \psfrag{Qmp}{$Q_{-+}$}
\psfrag{Ione}{$\quad\mathbf{1}$}
\psfrag{Xiota}{$\iota$}\psfrag{Xiotap}{$\iota'$}
\psfrag{Xgi}{$g_i$}\psfrag{Xgim1}{$g_i^{-1}$}
\psfrag{Xsumi}{\large{$\sum_{i\in I}$}}
\psfrag{Xbeta}{$\overline{\beta}$}
\psfrag{Xgj}{$g_j$} \psfrag{ifinotj}{\large{if $i\not= j$}}
\psfrag{XSnLm}{$S^n_-\otimes \Lambda^m_+$}
\psfrag{XLmSn}{$\Lambda^m_+\otimes S^n_-$}
\psfrag{XLm1Sn1}{$\Lambda^{m-1}_+\otimes S^{n-1}_-$}
\psfrag{XSumkb}{$-\sum_{b=0}^{k-2} (k-b-1)$}

\title{Categorification of the Heisenberg Algebra and MacMahon Function}

\date{}

\begin{document}

\author{Na Wang}
\author{Zhixi Wang}
\author{Ke Wu}
\author{Jie Yang}
\author{Zifeng Yang}

\address{Na Wang\\College of mathematics and information sciences
\\Henan University, Kaifeng, 475001, China\\}
\email{wangnawanda@163.com}

\address{Zhixi Wang\\School of mathematical science\\ Capital Normal University,
Beijing, 100048, China\\}
\email{wangzhx@mail.cnu.edu.cn}

\address{Ke Wu\\School of mathematical science\\Capital Normal University,
Beijing, 100048, China\\}
\email{wuke@mail.cnu.edu.cn}

\address{Jie Yang\\School of mathematical science\\Capital Normal University,
Beijing, 100048, China\\}
\email{yangjie@cnu.edu.cn}

\address{Zifeng Yang\\School of mathematical science\\Capital Normal University,
Beijing, 100048, China}
\email{yangzf@mail.cnu.edu.cn}

\begin{abstract}
Starting from a one dimensional vector space, we
construct a categorification $'\mathcal H$ of a deformed Heiserberg algebra $'H_{\Zt}$
by Cautis and Licata's method. The Grothendieck ring of $'\mathcal H$ is $'H_{\Zt}$.
As an application, we discuss some related partition functions related to
the MacMahon function of 3D Young diagram. We expect further applications of
the results of this paper.
\end{abstract}

\maketitle
\setcounter{page}{1}

\section{Introduction}

Let $H$ be the classical Heisenberg algebra of infinite rank, i.e., the algebra generated by the
sequence $\{a_n\}_{n\in{\Z}}$ with the defining relations
\begin{equation}\label{e:1.1}
a_na_m=a_ma_n+n{\delta}_{n+m,0}1, \quad n,m\in{\Z}.
\end{equation}
The generators $a_n$ for $n\in {\Z}$ can be realized from the fermions $\phi_n, \phi_n^* (n\in{\Z})$,
which satisfy the following relations
\begin{equation}\label{e:1.2}
[\phi_n,\phi_m]_+=0, \,\, [\phi_n^*,\phi_m^*]=0,\,\, [\phi_n^*,\phi_m]_+=\delta_{n+m,0},
\end{equation}
where we use the notation $[X,Y]_+=XY+YX$. Let ${\mathcal A}$ be the Clifford algebra generated by the
$\phi_n,\phi^*_n, n\in{\Z}$. Then an element of ${\mathcal A}$ can be written as a finite linear
combination of monomials of the form
\[
\phi_{m_1}\cdots\phi_{m_r}\phi^*_{n_1}\cdots\phi^*_{n_s}, \quad {\text{where }} m_1<\cdots<m_r, \,\,
n_1<\cdots <n_s.
\]
By introducing a variable $k$, we define the fermionic generating functions as the formal
sums
\[
\phi(k)=\sum\limits_{n\in{\Z}+\frac{1}{2}} \phi_n k^{-n-\frac{1}{2}}, \quad
\phi^*(k)=\sum\limits_{n\in{\Z}+\frac{1}{2}}\phi^*_nk^{-n-\frac{1}{2}}.
\]
Then the $a_n$ can be realized by the following equation
\begin{equation}\label{e:fb}
\sum\limits_{n\in{\Z}}a_n k^{-n-1}=:\phi(k)\phi^*(k):
\end{equation}
where $::$ is the normal order, see \cite{MJD} for detail.

The integral form $H_{\Z}$ of Heisenberg algebra $H$ using the vertex operators in terms of $a_n$
is generated over the ring $\Z$ by $\{p_n,q_n\}_{n\in{\mathbb N}}$ with relations
\begin{eqnarray}
q_n p_m & = & \sum_{k\geq 0}p_{m-k} q_{n-k}\label{qp}, \label{e:1.3}\\
 q_n q_m & = & q_m q_n, \label{e:1.4}\\
 p_n p_m & = & p_m p_n, \label{e:1.5}\label{pq}
\end{eqnarray}
in which $p_{0} = q_{0} = 1$ and $p_{k} = q_{k} = 0$ for $k <0$,
thus the summation in Equation (\ref{e:1.3}) is finite.
These generators $q_n, \ p_m$ are obtained as the homogeneous
components in $z$ of the ¡°halves of vertex operators¡±
\begin{equation}\label{e:pqrelation}
\sum\limits_{m\ge 0} q_m z^{-m}=\exp(\sum_{m\geq 1} \frac{1}{m}a_{m} z^{-m}), \qquad \text{and} \qquad
\sum\limits_{m\ge 0} p_m z^m = \exp(\sum_{m\geq 1} \frac{1}{m}a_{-m} z^m),
\end{equation}
which play a very important role in the conformal field theory and the theory of quantum
algebras.

The authors use slightly different presentations of integral forms of Heisenberg algebra in literature.
Khovanov used an integral form $H'_{\mathbb{Z}}$ of Heisenberg algebra in \cite{K}.
He introduced a calculus of planar diagrams for biadjoint functors and degenerate affine
Hecke algebras, which led to an additive monoidal category whose Grothendieck ring contains the
integral form $H'_{\mathbb{Z}}$ (and Khovanov conjectured that they are isomorphic).

We define the deformed integral form $'H_{\Zt}$ of Heisenberg algebra as the algebra
generated by $\{p_n, q_n\}_{n\in{\mathbb N}}$ over the ring ${\Z}[t,t^{-1}]$ with relations
\begin{eqnarray}
q_n p_m & = & \sum_{k\geq 0}[k+1]p_{m-k} q_{n-k}, \label{qp}\\
 q_n q_m & = & q_m q_n, \\
 p_n p_m & = & p_m p_n,\label{pq}
\end{eqnarray}
where $p_{0} = q_{0} = 1$, \, $p_{k} = q_{k} = 0$ for $k <0$,
and the quantum integer $[k]=\frac{t^{k}-1}{t-1}$. Under the relation (\ref{e:pqrelation})
between $\{p_n,q_n\}_{n\ge 0}$ and $\{a_n\}_{n\in \Z, n\ne 0}$ when the coefficients
are extended to ${\C}[t,t^{-1}]$, we can find that the deformed Heisenberg algebra $'H=\,'H_{\Ct}$
is generated by $\{a_n\}_{n\in\Z}$ with the defining relation
\begin{equation}\label{e:arelation}
[a_n, a_m]=n(1+t^n)\delta_{n+m, 0}\, 1.
\end{equation}

It is important to lift the Heisenberg algebra (any form) to some categoric version in representation theory.
Cautis and Licata \cite{CL} categorified the Heisenberg algebras ${\mathfrak h}_{\Gamma}$ for some finite subgroup
${\Gamma}$ of $SL_2(\C)$, but the discussion in that paper clearly can not be applied to the case when $\Gamma$ is the
trivial group. We will give a categorification of the Heisenberg algebra $'H_{\Zt}$ by the same techinique, and derive the MacMahon function from this categorification.  The results of this paper are
expected to be applied in the study of topics in the theory of quantum algebras which
we are undergoing.

The authors are grateful to Morningside Center of Chinese
Academy of Sciences for providing excellent research environment and financial
support. This work is also partially supported by NSF grant 11031005 and grant KZ201210028032.

\section{Categorification of the deformed Heisenberg algebra $'H_{\Zt}$}\label{section:2}
Let $V=\C v$ be a 1-dimensional vector space over the field $\C$ of complex numbers,
$\Lambda^*(V)$ the exterior algebra of $V$ with the basis $B$ consisting of  $1, \, v$. We define a $\C$-linear map $tr:\ \Lambda^*(V)\rightarrow \C$ by setting
\[
tr (v)=1,\ \ \ tr(1)=0.
\]
Let $B^\vee$ be the basis of $\Lambda^*(V)$ dual to $B$ with respect to the non-degenerate bilinear form $\langle a,b \rangle := \tr(ab)$. Under this bilinear form, we see that $1^\vee=v,\ v^\vee=1$.
$\Lambda^*(V)$ is graded as $\Lambda^*(V)={\C}\oplus {\C}v$, with the degree of
$b\in \Lambda^*(V)$ denoted by $|b|$, and $|\alpha|=0$ for $0\ne \alpha\in {\C}$, $|v|=1$.

We define a $2$-category $'\H'$ as follows.
\begin{itemize}
\item {\bf Objects (that is, $0$-cells) of $'\H'$.} The objects are the integers.
\item {\bf $1$-morphisms (that is, $1$-cells) of $'\H'$.}
Let $P:\ n\longmapsto n+1$ and $Q:\ n+1\longmapsto n$ be
two fundamental $1$-morphisms, $\mathbf 1$ the identity $1$-morphism on the objects,
and the $1$-morphisms $P\lba l\rba$, $Q\lba l\rba$, and ${\mathbf 1}\lba l\rba$ for
integers $l$ which we call the shifted $1$-morphisms by dimension $l$. The $1$-morphisms
of $'\H'$ are generated from $P\lba l\rba, Q\lba l\rba$, and $\mathbf 1\lba l\rba$ by finite
compositions and finite direct sums. The dimension shifting on the $1$-morphisms satisfy
the equality $A\lba k\rba\cdot B\lba l\rba=A\cdot B\lba k+l\rba$ for $1$-morphisms $A$ and $B$.
Note that $P\lba 0\rba, Q\lba 0\rba, {\mathbf 1}\lba 0\rba$ and $P, Q, {\mathbf 1}$ are
respectively the same.
\item {\bf $2$-morphisms (that is, $2$-cells) of $'\H'$.} The $2$-morphisms constitute a
$\C$-vector space. It is generated over $\C$
by planar oriented diagrams modulo relative
to boundary isotopies and some local relations which are described in detail in the following
paragraphs. The $2$-morphisms are also graded and the grading is compatible with the
grading of $1$-morphisms.
\end{itemize}
The detailed description is given in the following paragraphs.

An upward oriented strand denotes the identity 2-morphism $\id: P \rightarrow P$ while a downward oriented strand denotes the identity 2-morphism $\id: Q \rightarrow Q$. Upward-oriented lines and downward-oriented lines carrying dots  labeled by elements $b \in \Lambda^*(V)$ are other 2-morphisms of $'\H'$, for example,
\begin{equation*}
\begin{tikzpicture}[>=stealth,baseline=25pt]
\draw (0,0) -- (0,1)[->];
\filldraw (0,0.5) circle (2pt)+(0.25,0) node {$b$};
\draw [shift={+(1.95,0.5)}](0,0) node{$\in {\rm Hom}_{'\H'}(P,P\lba |b| \rba),$};
\draw [shift={+(4.5,0)}](0,0) -- (0,1)[<-];
\filldraw [shift={+(4.5,0.67)}](0,0) circle (2pt)+(0.25,0) node {$b'$};
\filldraw [shift={+(4.5,0)}](0,0.4) circle (2pt)+(0.25,0.4) node {$b''$};
\draw [shift={+(6.75,0.5)}](0,0) node{$\in {\rm Hom}_{'\H'}(Q,Q\lba |b'b''|\rba),$};
\draw [shift={+(9.2,0)}](0,0) -- (1,1)[<-];
\draw [shift={+(9.2,0)}](1,0) -- (0,1)[->];
\filldraw [shift={+(9.2,0)}](0.3,0.3) circle (2pt)+(0.45,0.4) node {$b$};
\draw [shift={+(12,0.5)}](0,0) node {$\in {\rm Hom}_{'\H'}(QP,PQ\lba |b|\rba).$};
\end{tikzpicture}
\end{equation*}
The dimension shifting appears in the above examples to be compatible with the grading of
the planar string diagrams, which is given later in this section. If no confusion occurs, we
also omit the dimension shifting notation besides a $1$-morphism.
On planar diagrams, compositions of 2-morphisms are depicted upward from bottom to top.
All these are similar to \cite{CL}, so are the local relations which the planar diagrams are
requested to satisfy. The only difference is that we adopt the
notation of dimension shifting $\lba \,\,\,\rba$
instead of the degree shifting $\la\,\,\,\ra$ in \cite{CL}.
for the $1$-morphisms. The dimension shifting is essential in discussion,
for reasons which can be seen better in Section
\ref{section:3} about a natural representation of the Karoubi envelope $'\H$ of $'\H'$.
We also understand that in a planar string diagram, the $1$-morphisms in the bottom are generated
by $P$, $Q$, and ${\mathbf 1}$ through finite compositions (that is, the $1$-morphisms without dimension
shifting), unless it is explicitly indicated otherwise.

The local relations are the following. First, the dots can move freely along strands and through intersections, for
example,
\begin{equation*}
\begin{split}
&\begin{tikzpicture}[>=stealth, baseline=25pt]
\draw [shift={+(0,0)}](0,0) -- (1,1)[<-];
\draw [shift={+(0,0)}](1,0) -- (0,1)[->];
\filldraw [shift={+(0,0)}](0.3,0.3) circle (2pt)+(0.45,0.4) node {$b$};
\draw [shift={+(1.85,0.5)}] node{ $=$};
\draw [shift={+(2.7,0)}](0,0) -- (1,1)[<-];
\draw [shift={+(2.7,0)}](1,0) -- (0,1)[->];
\filldraw [shift={+(2.7,0)}](0.8,0.8) circle (2pt)+(0.95,0.75) node {$b$};
\draw [shift={+(3.7,0)}] (0.3,0) node{,};
\draw [shift={+(6,0)}](0,0) -- (1,1)[<-];
\draw [shift={+(6,0)}](1,0) -- (0,1)[->];
\filldraw [shift={+(6,0)}](0.2,0.8) circle (2pt)+(0.35,0.9) node {$b$};
\draw [shift={+(7.85,0.5)}] node{ $=$};
\draw [shift={+(8.7,0)}](0,0) -- (1,1)[<-];
\draw [shift={+(8.7,0)}](1,0) -- (0,1)[->];
\filldraw [shift={+(8.7,0)}](0.8,0.2) circle (2pt)+(0.95,0.3) node {$b$};
\draw [shift={+(9.7,0)}] (0.3,0) node{,};
\end{tikzpicture}\\
&\begin{tikzpicture}[>=stealth]
\draw (0,0.75) arc (180:360:.5);
\draw (0,1.2) -- (0,0.75) ;
\draw (1,1.2) -- (1,0.75) [<-];
\filldraw  (0,0.9) circle (2pt)+(.25,0) node {$b$};
\draw [shift={+(2,0.6)}](0,0) node{$=$};
\draw [shift={+(3,0)}](0,0.75) arc (180:360:.5);
\draw [shift={+(3,0)}](0,1.2) -- (0,0.75) ;
\draw [shift={+(3,0)}](1,1.2) -- (1,0.75) [<-];
\filldraw [shift={+(3,0)}] (0.5,0.25) circle (2pt)+(.75,0.25) node {$b$};
\draw [shift={+(5,0.6)}](0,0) node{$=$};
\draw [shift={+(6,0)}](0,0.75) arc (180:360:.5);
\draw [shift={+(6,0)}](0,1.2) -- (0,0.75) ;
\draw [shift={+(6,0)}](1,1.2) -- (1,0.75) [<-];
\filldraw [shift={+(6,0)}] (1,0.9) circle (2pt)+(1.25,0.9) node {$b$};
\draw [shift={(6,0)}] (1.5, 0.4) node{.};
\end{tikzpicture}
\end{split}
\end{equation*}
%
%
%
%
Second, collision of dots is controlled by the multiplication in the exterior algebra ${\Lambda}^*(V)$:
\begin{equation*}
\begin{tikzpicture}[>=stealth]
\draw (0,0) -- (0,1)[->];
\filldraw (0,0.5) circle (2pt)+(-0.35,0) node {$b_1b_2$};
\draw [shift={+(0.8,0.5)}](0,0) node{$=$};
\draw [shift={+(1.6,0)}](0,0) -- (0,1)[->];
\filldraw [shift={+(1.6,0)}] (0,0.7) circle (2pt)+(0.25,0.7) node {$b_1$};
\filldraw [shift={+(1.6,0)}] (0,0.3) circle (2pt)+(0.25,0.3) node {$b_2$};
\draw [shift={+(2.6,0)}](0,0) node{$,$};

\draw [shift={+(5.6,0)}] (0,0) -- (0,1)[<-];
\filldraw [shift={+(5.6,0)}] (0,0.5) circle (2pt)+(-0.35,0.5) node {$b_1b_2$};
\draw [shift={+(6.4,0.5)}](0,0) node{$=$};
\draw [shift={+(7.2,0)}](0,0) -- (0,1)[<-];
\filldraw [shift={+(7.2,0)}] (0,0.7) circle (2pt)+(0.25,0.7) node {$b_1$};
\filldraw [shift={+(7.2,0)}] (0,0.3) circle (2pt)+(0.25,0.3) node {$b_2$};
\draw [shift={+(7.2,0)}](1,0) node{$.$};
\end{tikzpicture}
\end{equation*}
Third, dots on strands supercommute when they move past one another, for example:
\begin{equation*}
\begin{tikzpicture}[>=stealth]
\draw (0,0) -- (0,1)[->];
\filldraw (0,0.3) circle (2pt)+(-0.25,0) node {$b_1$};
\draw [shift={+(0.6,0.5)}](0,0) node{$\cdots$};
\draw [shift={+(1.2,0)}](0,0) -- (0,1)[->];
\filldraw [shift={+(1.2,0)}] (0,0.7) circle (2pt)+(0.25,0.7) node {$b_2$};
\draw [shift={+(2.2,0)}](0,0.5) node{$=$};
\draw [shift={+(3.7,0)}] (0,0.5) node {$(-1)^{|b_1||b_2|}$};
\draw [shift={+(5.2,0)}] (0,0) -- (0,1)[->];
\filldraw [shift={+(5.2,0)}] (0,0.7) circle (2pt)+(-0.25,0.7) node {$b_1$};
\draw [shift={+(5.8,0.5)}](0,0) node{$\cdots$};
\draw [shift={+(6.4,0)}](0,0) -- (0,1)[->];
\filldraw [shift={+(6.4,0)}] (0,0.3) circle (2pt)+(0.25,0.3) node {$b_2$.};
\end{tikzpicture}
\end{equation*}
Last, local relations contain the following:
\begin{equation}\label{local1}
\begin{tikzpicture}[>=stealth,baseline=0.8cm]
  \draw (0,0) .. controls (0.8,0.8) .. (0,1.6)[->];
  \draw (0.8,0) .. controls (0,0.8) .. (0.8,1.6)[->] ;
  \draw (1.2,0.8) node {=};
  \draw (2,0) --(2,1.6)[->];
  \draw (2.8,0) node{\qquad ,} -- (2.8,1.6)[->];
  \draw [shift={+(4.8,0)}](0,0) -- (1.6,1.6)[->];
  \draw [shift={+(4.8,0)}](1.6,0) -- (0,1.6)[->];
  \draw [shift={+(4.8,0)}](0.8,0) .. controls (0,0.8) .. (0.8,1.6)[->];
  \draw [shift={+(4.8,0)}](2,0.8) node {=};
  \draw [shift={+(4.8,0)}](2.4,0) -- (4,1.6)[->];
  \draw [shift={+(4.8,0)}](4,0) node{\qquad ;} -- (2.4,1.6)[->];
  \draw [shift={+(4.8,0)}](3.2,0) .. controls (4,0.8) .. (3.2,1.6)[->];
\end{tikzpicture}
\end{equation}

\begin{equation}
\begin{tikzpicture}[>=stealth,baseline=0.8cm]
\draw (0,0) .. controls (0.8,0.8) .. (0,1.6)[<-];
\draw (0.8,0) .. controls (0,0.8) .. (0.8,1.6)[->];
\draw (1.2,0.8) node {=};
\draw (1.6,0) --(1.6,1.6)[<-];
\draw (2.4,0) -- (2.4,1.6)[->];
\draw (2.8,0.8) node {$-$};
\draw (3.2,1.4) arc (180:360:.4);
\draw (3.2,1.6) -- (3.2,1.4) ;
\draw (4,1.6) -- (4,1.4) [<-];
\draw (4,.2) arc (0:180:.4) ;
\filldraw  (3.6,1) circle (2pt)+(.2,0) node {$v$};
\draw (4,0) -- (4,.2) ;
\draw (3.2,0) -- (3.2,.2) [<-];
\draw (4.4,0.8) node {$-$};
\filldraw  (5.2,0.6) circle (2pt)+(.2,0) node {$v$};
\draw (4.8,1.4) arc (180:360:.4);
\draw (4.8,1.6) -- (4.8,1.4) ;
\draw (5.6,1.6) -- (5.6,1.4) [<-];
\draw (5.6,.2) arc (0:180:.4);
\draw (5.6,0) -- (5.6,.2) node{\qquad ;} ;
\draw (4.8,0) -- (4.8,.2) [<-];
\end{tikzpicture}
\end{equation}

\begin{equation}\label{local2}
\begin{tikzpicture}[>=stealth,baseline=0.8cm]
\draw (0,0) .. controls (0.8,0.8) .. (0,1.6)[->];
\draw (0.8,0) .. controls (0,0.8) .. (0.8,1.6)[<-] ;
\draw (1.2,0.8) node {=};
\draw (1.84,0) --(1.84,1.6)[->];
\draw (2.64,0) node{\qquad ,} -- (2.64,1.6)[<-];

\draw [shift={+(4.8,0.8)}](-0.8,0) .. controls (-0.8,.4) and (-.24,.4) .. (-.08,0) ;
\draw [shift={+(4.8,0.8)}](-0.8,0) .. controls (-0.8,-.4) and (-.24,-.4) .. (-.08,0) ;
\draw [shift={+(4.8,0.8)}](0,-0.8) .. controls (0,-.4) .. (-.08,0) ;
\draw [shift={+(4.8,0.8)}](-.08,0) .. controls (0,.4) .. (0,0.8) [->] ;
\draw [shift={+(4.8,0.8)}](0.56,0) node {$=0$\quad ;};
\end{tikzpicture}
\end{equation}

\begin{equation}
\begin{tikzpicture}[>=stealth,baseline=0cm]
\draw [shift={+(0,0)}](0,0) arc (180:360:0.5cm) ;
\draw [shift={+(0,0)}][->](1,0) arc (0:180:0.5cm) ;
\filldraw [shift={+(1,0)}](0,0) circle (2pt);
\draw [shift={+(0,0)}](1,0) node [anchor=east] {$b$};
\draw [shift={+(0,0)}](1.75,0) node{$\qquad = \tr(b)$\,\,.};
\end{tikzpicture}
\end{equation}
We notice that the above local relations imply that the second relation of (\ref{local1}) is
satisfied when any of the arrows of the strings in the diagram on the
left hand side of the equality is reversed (and the arrows of the strings of the right
hand side of the equality are reversed correspondingly), hence we
depict this local relation as:
\begin{equation}\label{local1a}
\begin{tikzpicture}[>=stealth,baseline=0.8cm]
  \draw [shift={+(4.8,0)}](0,0) -- (1.6,1.6)[-];
  \draw [shift={+(4.8,0)}](1.6,0) -- (0,1.6)[-];
  \draw [shift={+(4.8,0)}](0.8,0) .. controls (0,0.8) .. (0.8,1.6)[-];
  \draw [shift={+(5,0)}](2,0.8) node {=};
  \draw [shift={+(5.2,0)}](2.4,0) -- (4,1.6)[-];
  \draw [shift={+(5.2,0)}](4,0) node{\qquad ,} -- (2.4,1.6)[-];
  \draw [shift={+(5.2,0)}](3.2,0) .. controls (4,0.8) .. (3.2,1.6)[-];
\end{tikzpicture}
\end{equation}
and first relation of (\ref{local1}) is satisfied when all the arrows of the strings
in the diagram are reversed.

We define the degrees of the planar string diagrams (as morphisms from
$X_1X_2\cdots X_k\to Y_1Y_2\cdots Y_l$ with $k\ge 0, l\ge 0$ integers,
$X_i,Y_j\in\{P,Q\}$, and $X_1X_2\cdots X_k={\mathbf 1}$ if $k=0$) from the degrees of
the simplest planar string diagrams:
\[
\begin{split}
&\begin{tikzpicture}[>=stealth]
\draw  (-.5,.5) node {$deg$};
\draw [->](0,0) -- (1,1);
\draw [->](1,0) -- (0,1);
\draw [shift={+(2,0)}] (-0.5,0.5) node {$=$};
\draw [shift={+(2.5,0)}] (0,.5) node {$deg$};
\draw [shift={+(3.0,0)}][<-](0,0) -- (1,1);
\draw [shift={+(3.0,0)}][<-](1,0) -- (0,1);
\draw [shift={+(5,0)}] (-0.5,0.5) node {$=$};
\draw [shift={+(5.5,0)}] (0,.5) node {$deg$};
\draw [shift={+(6.0,0)}][->](0,0) -- (1,1);
\draw [shift={+(6.0,0)}][<-](1,0) -- (0,1);
\draw [shift={+(8,0)}] (-0.5,0.5) node {$=$};
\draw [shift={+(8.5,0)}] (0,.5) node {$deg$};
\draw [shift={+(9.0,0)}][<-](0,0) -- (1,1);
\draw [shift={+(9.0,0)}][->](1,0) -- (0,1);
\draw [shift={+(9.0,0)}](1.5,.5) node{$ \quad = 0 \quad ;$};
\end{tikzpicture}\\
&\begin{tikzpicture}[>=stealth]
\draw  (-.5,-.25) node {$deg$};
\draw (0,0) arc (180:360:.5)[->] ;
\draw (1.75,-.25) node{$ \qquad =-1 \,\, ,\quad $};
\draw (3,-.25) node{$deg$};
\draw (4.5,-.5) arc (0:180:.5) [->];
\draw (5,-.25) node{$ \quad =0 \quad ;$};
\end{tikzpicture}\\
&\begin{tikzpicture}[>=stealth]
\draw  (-.5,-.25) node {$deg$};
\draw (0,0) arc (180:360:.5)[<-] ;
\draw (1.75,-.25) node{$ \quad = 0 \quad , \quad $};
\draw (3,-.25) node{$deg$};
\draw (4.5,-.5) arc (0:180:.5) [<-];
\draw (5,-.25) node{$ \quad = 1 \quad ; $};
\end{tikzpicture}
\end{split}
\]
and the degree of a dot labeled by $b$ equals the degree of $b$ in the graded algebra $\Lambda^*(V)$.
If $f$ and $g$ are two planar string diagrams, and they can be composed as $g\circ f\ne 0$, then we
define $deg(g\circ f):= deg(f) + deg(g)$. The degree of a string diagram is defined to be
the total sum of the degrees of all the planar strings in the string diagram. And finally,
if $f=\sum_i c_i f_i$ with each $f_i\ne 0$ a string diagram and $0\ne c_i\in\C$, then we define
\[
deg(f)=\max_i\{deg(f_i)\}.
\]

The $2$-morphisms of $'\H'$ constitute a ring with the obvious addition and the multiplication given
by the composition of planar string diagrams. This ring is also a graded $\C$-vector space. We define
\[
f\cdot g=0
\]
if the two planar string diagrams $f$ and $g$ can not be composed.

Be caution that in the above definition of $deg$, the notation $deg(f)$ is defined only for
$f=\sum_i c_i f_i$ with $c_i\in \C$ and those $f_i$ being the planar string diagrams connecting
the $1$-morphisms generated by $P$, $Q$, and $\mathbf 1$, without dimension shiftings.
For a planar string diagram $f: A\lra B$ (so $A$ and $B$ are finite compositions of $P$, $Q$, and
$\mathbf 1$), we make a convention that $f$ is a $2$-morphism in $\HHhom(A\lba k\rba, B\lba l\rba)$
for any $k,l\in \Z$. Therefore the $\C$-vector space $\HHhom(A, B)$ is graded, with
$\HHhom(A,B)_0$ denoting by the space of $2$-morphisms of degree $0$. To be more precise,
assume $f\in N={\rm Hom}_{'\H'}(X,Y)$, where
$X=X_1X_2\cdots X_{k_1}$ and $Y=Y_1Y_2\cdots Y_{k_2}$ with $X_i, Y_j\in \{P,Q\}$
(that is, $f$ is a planar string diagram which connects the sequence of $X_i$'s on the
bottom and the sequence of $Y_i$'s on the top), we also
view $f\in M={\rm Hom}_{'\H'}(X\lba l_1\rba, Y\lba l_2\rba)$. We then define
\[
deg_{M}(f)=deg_{N}(f)+l_1-l_2=deg(f)+l_1-l_2.
\]

In addition, we define the shifted degree (which is denoted by $sdeg$) of the planar diagrams
(also viewed as $2$-morphisms from $X_1X_2\cdots X_k\to Y_1Y_2\cdots Y_l$ with
$X_i,Y_j\in \{P,Q\}$):
\[
\begin{split}
&\begin{tikzpicture}[>=stealth]
\draw  (-.5,.5) node {$sdeg$};
\draw [->](0,0) -- (1,1);
\draw [->](1,0) -- (0,1);
\draw [shift={+(2,0)}] (-0.5,0.5) node {$=$};
\draw [shift={+(2.5,0)}] (0,.5) node {$sdeg$};
\draw [shift={+(3.0,0)}][<-](0,0) -- (1,1);
\draw [shift={+(3.0,0)}][<-](1,0) -- (0,1);
\draw [shift={+(5,0)}] (-0.5,0.5) node {$=$};
\draw [shift={+(5.5,0)}] (0,.5) node {$sdeg$};
\draw [shift={+(6.0,0)}][->](0,0) -- (1,1);
\draw [shift={+(6.0,0)}][<-](1,0) -- (0,1);
\draw [shift={+(8,0)}] (-0.5,0.5) node {$=$};
\draw [shift={+(8.5,0)}] (0,.5) node {$sdeg$};
\draw [shift={+(9.0,0)}][<-](0,0) -- (1,1);
\draw [shift={+(9.0,0)}][->](1,0) -- (0,1);
\draw [shift={+(9.0,0)}](1.5,.5) node{$ \quad = 0 \quad ;$};
\end{tikzpicture}\\
&\begin{tikzpicture}[>=stealth]
\draw  (-.5,-.25) node {$sdeg$};
\draw (0,0) arc (180:360:.5)[->] ;
\draw (1.75,-.25) node{$ \qquad =-1 \,\, ,\quad $};
\draw (3,-.25) node{$sdeg$};
\draw (4.5,-.5) arc (0:180:.5) [->];
\draw (5,-.25) node{$ \quad =0 \quad ;$};
\end{tikzpicture}\\
&\begin{tikzpicture}[>=stealth]
\draw  (-.5,-.25) node {$sdeg$};
\draw (0,0) arc (180:360:.5)[<-] ;
\draw (1.75,-.25) node{$ \quad = 0 \quad , \quad $};
\draw (3,-.25) node{$sdeg$};
\draw (4.5,-.5) arc (0:180:.5) [<-];
\draw (5,-.25) node{$ \quad = 1 \quad ; $};
\end{tikzpicture}
\end{split}
\]
and the shifted degree $sdeg(b)$ of a dot labeled by $b \in \Lambda^*(V)$ is
the same as $deg(b)$ defined above. By convention, we set $deg(0)=sdeg(0)=+\infty$.

\begin{Remark}\label{rmk:1}
We have the inequalities
\[
deg(f g)\ge deg(f)+deg(g), \qquad sdeg(f g) \ge sdeg(f) + sdeg(g)
\]
for the 1-morphisms $f$ and $g$, but the equalities don't hold in general,
even the composition $f\circ g$ exists, for example, the dot
morphism labeled by $v$ has both the degree and the shifted degree equal to
1, but the square is equal to 0.
\end{Remark}

From the local relations, we get
\begin{equation}\label{e:basic}
QP\cong PQ \oplus \mathbf{1}\oplus \mathbf 1 \lba 1\rba
\end{equation}
This can be seen by the following diagram:
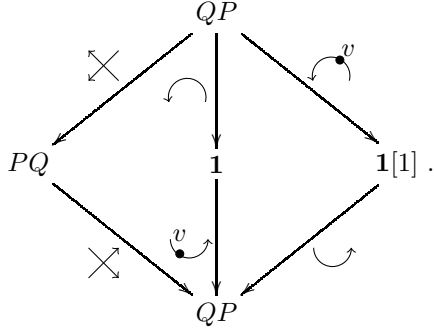
\[
\xy
(0,20)*+ {QP}="T1";
(-25,0)*+ {PQ}="m1";
(0,0)*+ {\mathbf 1}="m2";
(25,0)*+{{\mathbf 1}\lba 1\rba\,\, .}="m3";
(0, -20)*+{QP}="T2";
 {\ar@/_0pc/_{ \begin{tikzpicture} \draw [<-] (0,0) -- (0.4,0.4);
 \draw [->] (0.4,0) -- (0,0.4);\end{tikzpicture}
  \xy
 \endxy}
  "T1";"m1"};
 {\ar@/_0.0pc/_{ \begin{tikzpicture} \draw (0.3,-.5) arc (-15:195:.25) [->];\end{tikzpicture} \xy
  \endxy} "T1";"m2"};
 {\ar@/^0pc/^{ \begin{tikzpicture} \draw (0.3,-.5) arc (-15:195:.25) [->];
  \filldraw (0.175,-0.225) circle (1.5pt); \draw (0.29,-0.1) node{$v$}; \end{tikzpicture}
  \xy
  \endxy} "T1";"m3"};
 {\ar@/_0pc/_{  \begin{tikzpicture} \draw [->] (0,0) -- (0.4,0.4);
 \draw [<-] (0.4,0) -- (0,0.4);\end{tikzpicture}
  \xy
  \endxy} "m1";"T2"};
 {\ar@/_0.0pc/_{ \begin{tikzpicture} \draw [shift={+(0,0.3)}](0,0) arc (180:360:.25)[->];
  \filldraw (0.125,0.105) circle (1.5pt) +(0,0.25) node{$v$}; \end{tikzpicture}
 \xy
    \endxy} "m2";"T2"};
 {\ar@/^0pc/^{ \begin{tikzpicture} \draw [shift={+(0,0.3)}](0,0) arc (180:360:.25)[->];\end{tikzpicture}
  \xy
    \endxy} "m3";"T2"};
\endxy
\]

By (\ref{local1}), we see that upward oriented crossings satisfy the relations of symmetric group
$S_n$. We have a canonical homomorphism
\[
{\C}[S_n] \longrightarrow {\End}_{'\H'}(P^n)
\]
from the group algebra of $S_n$ to the endomorphism ring of the $n$-th tensor power of $P$.
Through some diagram manipulation according to $2$-morphisms of $'\H'$, we can see that the
relation (\ref{local1}) is also satisfied by the downward oriented crossings, therefore
get the similar canonical homomorphism
\[
{\C}[S_n] \longrightarrow {\End}_{'\H'}(Q^n).
\]
For a positive integer $n$ and a partition $\lambda$ of $n$, let $e_{\lambda}$ be the associated Young
symmetrizer of ${\C}[S_n]$. The Young symmetrizers are idempotents of the group algebra ${\C}[S_n]$ given
by
\[
e_{\lambda}=a_{\lambda} b_{\lambda}/{n_{\lambda}}
\]
where $n_{\lambda}={n!/{\rm dim}(V_{\lambda})}$ with $V_{\lambda}$ the irreducible representation
corresponding to the partition $\lambda$, the elements
$a_{\lambda}, b_{\lambda}\in {\C}[S_n]$ are defined for a Young tableau $T$ of the partition
$\lambda$ by
\[
a_{\lambda}=\sum\limits_{g\in L_{\lambda}} g, \qquad b_{\lambda}=\sum\limits_{g\in L'_{\lambda}} sign(g)g,
\]
and
\[
\begin{split}
&L_{\lambda}=\{g\in S_n: \, g\text{ preserves each row of $T$}\}, \\
&L'_{\lambda}=\{g\in S_n: \, g\text{ preserves each column of $T$}\}.
\end{split}
\]
The irreducible representations of $S_n$ are given by ${\C}[S_n]e_{\lambda}$ for partitions $\lambda$ of
$n$. Let $e_{(n)}$ be the idempotent corresponding to the partition $(n)$ of $n$, with
${\C}[S_n]e_{(n)}$ being the trivial representation of $S_n$. See \cite{FH} for detail.

Let $'\H$ be the Karoubi envelope of $'\H'$. This is a $2$-category described in the following.
\begin{itemize}
\item {\bf Objects of $'\H$:} same as the objects of $'\H'$.
\item {\bf $1$-morphisms of $'\H$:} pairs $(M,e)$, where $M$ is a $1$-morphism of $'\H'$, and
$e$ is an idempotent $2$-morphism $M\xrightarrow{e} M$ of $'\H'$.
\item {\bf $2$-morphisms of $'\H$:} given two $1$-morphisms $(M, e)$ and $(M',e')$, the set
${\rm Hom}((M,e), (M',e'))$ consists of the $2$-morphisms $f:M\to M'$ of $'\H'$ such that
the diagram
\[
\xymatrix{M\ar[r]^{f}\ar[d]_{e} & M'\ar[d]^{e'}\\
M\ar[r]_{f} & M'}.
\]
\end{itemize}

The $2$-morphisms of $'\H$ constitute a ring, as so do the $2$-morphisms of $'\H'$.

For convenience, we also view $'\H'$ and $'\H$ as categories with the $1$-morphisms considered
as the objects and the $2$-morphisms as the morphisms (hence forgetting the set $\Z$ of $0$-cells).
Then $'\H'$ is a subcategory of $'\H$, and both are $\C$-linear additive monoidal categories.
For an idempotent $e:M\to M$ in $'\H'$ with $M$ a $1$-morphism, ${\id}_{M}-e$ is also an
idempotent, and we have a direct sum decomposition in $'\H$:
\[
M=(M,e)\oplus (M,{\id}_{M}-e).
\]

The 1-morphisms $P^n$ and $Q^n$
splits into direct sums of $(P^n,e_{\lambda})$ and $(Q^n,e_{\lambda})$ in $'\H$ respectively, over
partitions $\lambda$ of $n$. Define
\[
P_{\lambda}=(P^n,e_{\lambda}), \qquad Q_{\lambda}=(Q^n,e_{\lambda}), \qquad
P_n = (P^n, e_{(n)}), \qquad  Q_n=(Q^n, e_{(n)}).
\]

\begin{prop}\label{prop:1}
For any non-negative integers $m_1,m_2,n_1,n_2$ and integers $l_1,l_2$, we have
\begin{itemize}
\item[(1).] ${\HHhom}(P^{m_1}Q^{n_1}\lba l_1\rba, P^{m_2}Q^{n_2}\lba l_2\rba)_0=0$ unless
$m_1-n_1=m_2-n_2$ and $l_2\ge l_1$;
\item[(2).] the shifted degree of an element of ${\HHhom}(P^{m_1}Q^{n_1}, P^{m_2}Q^{n_2})$ is always
greater than $0$ (therefore greater than or equal to $1/2$), provided with $(m_1,n_1)\ne (m_2,n_2)$;
\item[(3).] there are no negative degree (respectively, shifted degree) endomorphisms
of $P^{m_1} Q^{n_1}\lba l\rba$, and the algebra of degree (respectively, shifted degree) zero endomorphisms
of $P^{m_1} Q^{n_1}\lba l\rba$ is isomorphic to ${\C}[S_{m_1}]\otimes {\C}[S_{n_1}]$.
\end{itemize}
\end{prop}
\begin{proof}
The conclusions are straightforward by noticing that
a cup planar string to $P^{m_2}Q^{n_2}$ or a cap planar string from $P^{m_1}Q^{n_1}$
(the arrow has to be clockwise since the $P$ is on the left of $Q$)
has degree greater than or equal to $0$, and shifted degree greater than $0$.
\end{proof}

\begin{theorem}[\cite{CL}]\label{theorem:1}
In the $2$-category $'\mathcal{H}$, the $1$-morphisms $P_n$ and $Q_n$ satisfy
\begin{eqnarray}
& &Q_n P_m \cong \bigoplus_{k\geq 0}\bigoplus_{l=0}^k P_{m-k}Q_{n-k}\lba l\rba,\label{$1}\\
& &Q_n Q_m \cong Q_m Q_n,\label{$2}\\
& &P_n P_m \cong P_m P_n.\label{$3}
\end{eqnarray}
\end{theorem}
\begin{proof}
The proof of this theorem is same as that of \cite{CL} in the case of
${\dim}V=2$ and a non-trivial finite subgroup $\Gamma$ of ${\rm SL}_2(\C)$
acting on $V$. We make a couple of simple comments here. To prove $P_nP_m\cong P_mP_n$, let
$\begin{tikzpicture}[>=stealth, baseline=0]\draw (0,0) rectangle (1.6,0.36); \draw (0.8,0.18) node {$(n)$};
\end{tikzpicture}$ denote the $1$-morphism $P_n$, and $\varphi: P_nP_m\to P_mP_n$,
$\psi:P_mP_n\to P_nP_m$ given respectively by

$$
\begin{tikzpicture}[scale=.75][>=stealth]
\draw (1,2.25) node {$\varphi:=$};
\draw (2.25,1) rectangle (4.25,1.5);
\draw (3.25,1.25) node {$(n)$};
\draw (4.4,1) rectangle (6.4,1.5);
\draw (5.4,1.25) node {$(m)$};

\draw [shift = {+(0,2)}](2.25,1) rectangle (4.25,1.5);
\draw [shift = {+(0,2)}](3.25,1.25) node {$(m)$};
\draw [shift = {+(0,2)}](4.4,1) rectangle (6.4,1.5);
\draw [shift = {+(0,2)}](5.4,1.25) node {$(n)$};

\draw (3,1.5) -- (6,3) [->];
\draw (6,1.5) -- (3,3) [->];

\draw (9.0,2.25) node {and};

\draw [shift = {+(9.7,0)}](1,2.25) node {$\psi:=$};
\draw [shift = {+(10,0)}](2.25,1) rectangle (4.25,1.5);
\draw [shift = {+(10,0)}](3.25,1.25) node {$(m)$};
\draw[shift = {+(10,0)}] (4.4,1) rectangle (6.4,1.5);
\draw [shift = {+(10,0)}](5.4,1.25) node {$(n)$};

\draw[shift = {+(10,0)}] [shift = {+(0,2)}](2.25,1) rectangle (4.25,1.5);
\draw [shift = {+(10,0)}][shift = {+(0,2)}](3.25,1.25) node {$(n)$};
\draw [shift = {+(10,0)}][shift = {+(0,2)}](4.4,1) rectangle (6.4,1.5);
\draw [shift = {+(10,0)}][shift = {+(0,2)}](5.4,1.25) node {$(m)$};

\draw [shift = {+(10,0)}](3,1.5) -- (6,3) [->];
\draw [shift = {+(10,0)}](6,1.5) -- (3,3) [->];
\end{tikzpicture}\,\, ,
$$
where each upward oriented strand is an abbreviation of $m$ or $n$ upward
oriented strands accordingly. We then get from the first relation of
(\ref{local1}) that
\[
\begin{tikzpicture}[scale=.75][>=stealth]
\draw (0,3.25) node {$\psi\circ\varphi=$};
\draw (2.25,1) rectangle (4.25,1.5);
\draw (3.25,1.25) node {$(n)$};
\draw (4.4,1) rectangle (6.4,1.5);
\draw (5.4,1.25) node {$(m)$};
\draw [shift = {+(0,2)}](2.25,1) rectangle (4.25,1.5);
\draw [shift = {+(0,2)}](3.25,1.25) node {$(m)$};
\draw [shift = {+(0,2)}](4.4,1) rectangle (6.4,1.5);
\draw [shift = {+(0,2)}](5.4,1.25) node {$(n)$};
\draw (3,1.5) -- (6,3) [->];
\draw (6,1.5) -- (3,3) [->];

\draw [shift = {+(0,4)}](2.25,1) rectangle (4.25,1.5);
\draw [shift = {+(0,4)}](3.25,1.25) node {$(n)$};
\draw [shift = {+(0,4)}](4.4,1) rectangle (6.4,1.5);
\draw [shift = {+(0,4)}](5.4,1.25) node {$(m)$};
\draw [shift = {+(0,2)}](3,1.5) -- (6,3) [->];
\draw [shift = {+(0,2)}](6,1.5) -- (3,3) [->];

\draw [shift = {+(7,2)}](1,1.25) node {$=$};

\draw [shift = {+(8,0)}](2.25,1) rectangle (4.25,1.5);
\draw [shift = {+(8,0)}](3.25,1.25) node {$(n)$};
\draw [shift = {+(8,0)}](4.4,1) rectangle (6.4,1.5);
\draw [shift = {+(8,0)}](5.4,1.25) node {$(m)$};
\draw [shift = {+(8,4)}](2.25,1) rectangle (4.25,1.5);
\draw [shift = {+(8,4)}](3.25,1.25) node {$(n)$};
\draw [shift = {+(8,4)}](4.4,1) rectangle (6.4,1.5);
\draw [shift = {+(8,4)}](5.4,1.25) node {$(m)$};

\draw [shift = {+(8,0)}] (3.5,1.5) -- (3.5,5) [->];
\draw [shift = {+(8,0)}] (5.5,1.5) -- (5.5,5) [->];
\end{tikzpicture}
\]
by pulling the braid straight, thus $\psi\circ\varphi$ is the identity $2$-morphism
on $P_nP_m$, and so is $\varphi\circ\psi$ on $P_mP_n$. In the same way we can
see $Q_nQ_m\cong Q_mQ_n$. To see why (\ref{$1}) is true, we at first deduce from
(\ref{e:basic}) that
\begin{equation}\label{e:cate}
QP^n\cong P^nQ\oplus\left(P^{n-1}\oplus P^{n-1}\lba 1\rba\right)^{\oplus n}
\cong P^nQ\oplus \left(P^{n-1}\right)^{\oplus n}\oplus \left(P^{n-1}\lba 1\rba\right)^{\oplus n},
\end{equation}
therefore in the $2$-category $'\H$, we have
\[
QP_n=Q(P^n,e_{(n)})\cong P_nQ\oplus P_{n-1}\oplus P_{n-1}\lba 1\rba,
\]
since $e_{(n)}=\frac{1}{n!}\sum_{\sigma\in S_n} \sigma$, the terms
$\left(P^{n-1}\right)^{\oplus n}$ and $\left(P^{n-1}\lba 1\rba \right)^{\oplus n}$
of (\ref{e:cate}) are respectively projected onto $P_{n-1}$ and $P_{n-1}\lba 1\rba$ as diagonals
when $QP^n$ is projected onto $QP_n$. Then an induction on $m$ gives
\[
Q^mP_n\cong \bigoplus_{j\ge 0} \left(\bigoplus_{l=0}^jP_{n-j}Q^{m-j}\lba l\rba^{\oplus \binom{j}{l}}\right)^{\oplus \binom{m}{j}},
\]
we get (\ref{$1}) after we project $Q^m$ to $Q_m=(Q^m,e_{(m)})$ in the above equation.
\end{proof}

From Theorem \ref{theorem:1}, we get a natural homomorphism of algebras
\begin{equation}\label{e:pi}
\pi:\,\, 'H_{\Zt} \longrightarrow K_0('\H)
\end{equation}
by sending $p_m$ to $[P_m]$, $q_n$ to $[Q_n]$, and $t^l$ to $[{\mathbf 1}\lba l\rba]$, where
$[X]$ denotes the isomorphism class of the $1$-morphism $X$ of $'\H$.

\begin{theorem}\label{theorem:1.5}
$\pi$ is an isomorphism of algebras.
\end{theorem}
\begin{proof}
The proof of this theorem will be given at the end of Section \ref{section:4}.
\end{proof}

\section{A Representation of the Category $'\H$}\label{section:3}
We will give a Fock representation of the $'\H$ in this section.

Let $A_1={\Sym}^*(V^{\vee})$, which can be identified with the
polynomial ring ${\C}[v^{\vee}]$. For each positive
integer $n$, define
\begin{equation*}
A_n=\left({\Sym}^*({\Vv})\otimes {\Sym}^*({\Vv})\otimes \cdots \otimes {\Sym}^*({\Vv})\right)
\rtimes S_n
\end{equation*}
where the symmetric group $S_n$ acts by permuting the $n$ terms in the product. These
algebras inherit the natural grading from ${\Sym}^*({\Vv})$. We denote
by ${\mathcal C}_n=D({\mathbf A_n\gmod})$ the bounded derived
category graded (left) $A_n$-modules which are finite dimensional over $\C$ with
the action of $v^{\vee}$ by $0$.
An object of ${\mathcal C}_n$ is in fact a finite dimensional ${\C}[S_n]$-module
(also say $S_n$-module) viewed as a torsion $A_n$-module, and is usually
expressed by its free resolution in the derived category ${\mathcal C}_n =D({\mathbf A_n\gmod})$.
We have the maps
\begin{equation*}
{\C}  \mathop{\longrightarrow}\limits^{i} A_1 \mathop{\longrightarrow}\limits^{p} {\C}
\end{equation*}
with $i$ the natural inclusion and $p$ the projection which takes ${\Sym}^{>0}(V)$ to zero. Thus
each $\C$-module is also an $A_1$-module and vice versa.

From the embedding $S_n=S_n\times S_1\hookrightarrow S_{n+1}$, we have the natural inclusion $A_n\otimes A_1
\hookrightarrow A_{n+1}$. Let $P(n) := A_{n+1}  \otimes_{A_n  \otimes A_1 } (A_n  \otimes \C)$ be
an $(A_{n+1}, A_n)$-bimodule, and $(n)Q  := (A_n  \otimes \C) \otimes_{A_n  \otimes A_1 } A_{n+1}$ be
an $(A_n,A_{n+1})$-bimodule. We define
the functor $\fP(n): D({\mathbf A_n  \gmod}) \rightarrow D({\mathbf A_{n+1}  \gmod})$ by
\begin{equation}\label{e:p}
\fP(n)(\cdot) := P (n) \otimes_{A_n } (\cdot),
\end{equation}
and similarly, the functor $(n)\fQ: D({\mathbf A_{n+1}  \gmod}) \rightarrow D({\mathbf A_n  \gmod})$ by
\begin{equation}\label{e:q}
(n)\fQ(\cdot) := (n)Q  \otimes_{A_{n+1} } (\cdot),
\end{equation}
with the obvious definitions of these two functors on the corresponding morphisms of the
categories. Moreover, we define the dimension shifting functor ${\id}\lba 1\rba$ by
\begin{equation*}
\begin{array}{rcl}
{\id}\lba 1\rba(n):  D({\mathbf A_n\gmod}) & \lra & D({\mathbf A_n\gmod})\\
M & \longmapsto  & M[1]
\end{array}
\end{equation*}
where $M[1]$ is the complex shifted to the left by one
degree from the complex $M$ in $D({\mathbf A_n\gmod})$ and $\id$ is the
identity functor on the derived categories $D({\mathbf A_n\gmod})$.

In the following, we define the following natural transformations:
\begin{itemize}
\item[(1)] $X(1): \fP \rightarrow \fP$ and $X(1): \fQ \rightarrow \fQ$,
\item[(2)] $X(b): \fP \rightarrow \fP [1]\{-1\}$ and $X(b): \fQ \rightarrow \fQ [1]\{-1\}$ for any $b \in V$,
\item[(3)] $T: \fP \fP \rightarrow \fP \fP$, $T: \fQ \fQ \rightarrow \fQ \fQ$, $T: \fQ \fP \rightarrow \fP \fQ$, and $T: \fP \fQ \rightarrow \fQ \fP$,
\item[(4)] $\adj: \fQ \fP \rightarrow \id$ and $\adj: \fP \fQ \rightarrow \id [1]\{-1\}$, and
\item[(5)] $\adj: \id \rightarrow \fQ \fP [-1]\{1\}$ and $\adj: \id \rightarrow \fP \fQ$.
\end{itemize}
where $[\cdot]$ denotes the cohomological dimension shift and $\{\cdot\}$ the grading shift.
The degree shifting notation $\{\cdot\}$ is attached only to reflect the degree of the
natural transformation, it doesn't impose impact on the functors generated by
$\fP$, $\fQ$, and $\id$.  The dimension
shifting operator $\lba \,\cdot\,\rba$ affects the Grothendieck ring of the category $'\H$.

{\bf (1). $X(1): \fP \rightarrow \fP$ and $X(1): \fQ \rightarrow \fQ$} are the identity natural transformations in $\fP$ and
$\fQ$ respectively.

{\bf (2). The definitions of $X(b): \fP \rightarrow \fP [1]\{-1\}$ and $X(b): \fQ \rightarrow \fQ [1]\{-1\}$ for any $b \in V$.} Viewed as an $A_1$-module, a free resolution of $\C$ is:
\begin{equation}\label{e:kres1}
\begin{array}{cccccccc}
0 & \mathop{\longrightarrow} &A_1\otimes V^\vee & \mathop{\longrightarrow}\limits^{d} &A_1 &\longrightarrow & {\C} \longrightarrow &0\\
  &              & f\otimes w &\longmapsto                 &fw. & & &
\end{array}
\end{equation}
We define a map $\phi_b$ by
\begin{equation*}
\begin{array}{rccc}
{\phi}_b: &A_1 \otimes V^{\vee} & \longrightarrow & A_1\{-1\} \\
&f\otimes w        & \longmapsto & \la w,b\ra f
\end{array}
\end{equation*}
where $\{-1\}$ denotes the degree shifting in $A_1={\Sym}^*(V^{\vee})$. Then we have the
following commutative diagram:
\begin{equation*}
\xymatrix{
         0\ar[r]\ar[d]              & A_1\otimes V^{\vee}\ar[r]^{d}\ar[d]^{\phi_b} &A_1\ar[r]\ar[d]  &0\ar@{=}[d]\\
A_1\otimes V^{\vee}\{-1\}\ar[r]^{-d} & A_1\{-1\}\ar[r]                     &0 \ar[r]         &0
}
\end{equation*}
which defines a morphism $X(b): {\C}\longrightarrow {\C}[1]\{-1\}$ in the derived category
$D({\mathbf A_1\gmod})$. This morphism then induces a morphism
\begin{equation*}
\xymatrix{
X(b):& P(n)\ar[r]\ar@{=}[d]     & P(n)[1]\{-1\} \ar@{=}[d]\\
&A_{n+1}\otimes_{A_n\otimes A_1}(A_n\otimes {\C}) & A_{n+1}\otimes_{A_n \otimes A_1}(A_n\otimes {\C})[1]\{-1\}
}
\end{equation*}
Thus we get the natural transformation $X(b): \fP\longrightarrow \fP[1]\{-1\}$. In the same way, we
apply the following free resolution of $\C$:
\begin{equation}\label{e:kres2}
\begin{array}{cccccccc}
0 & \mathop{\longrightarrow} &V^{\vee}\otimes A_1  & \mathop{\longrightarrow}\limits^{d} &A_1 &\longrightarrow & {\C} \longrightarrow &0\\
  &              & w\otimes f &\longmapsto                 &wf & & &
\end{array}
\end{equation}
and define a map $\psi_b$ by
\begin{equation*}
\begin{array}{rccc}
{\psi}_b: &V^{\vee}\otimes A_1 & \longrightarrow & A_1\{-1\} \\
&w\otimes f        & \longmapsto & -\la w,b\ra f
\end{array}
\end{equation*}
Then the following commutative diagram
\begin{equation*}
\xymatrix{
         0\ar[r]\ar[d]              & V^{\vee}\otimes A_1\ar[r]^{d}\ar[d]^{\psi_v} &A_1\ar[r]\ar[d]  &0\ar@{=}[d]\\
V^{\vee}\otimes A_1\{-1\}\ar[r]^{-d} & A_1\{-1\}\ar[r]                     &0 \ar[r]         &0
}
\end{equation*}
gives a morphism $X(b): {\C}\longrightarrow {\C}[1]\{-1\}$ in the derived category
$D({\mathbf A_1\gmod})$ and induces a morphism
\begin{equation*}
\xymatrix{
X(b):& (n)Q\ar[r]\ar@{=}[d]     & (n)Q[1]\{-1\} \ar@{=}[d]\\
&(A_n\otimes {\C})\otimes_{A_n\otimes A_1} A_{n+1} & (A_n\otimes {\C})\otimes_{A_n\otimes A_1} A_{n+1}[1]\{-1\}
}
\end{equation*}
Thus we get the natural transformation $X(b): \fQ\longrightarrow {\fQ}[1]\{-1\}$.

{\bf (3). The definitions of  $T: \fP \fP \rightarrow \fP \fP$, $T: \fQ \fQ \rightarrow \fQ \fQ$, $T: \fQ \fP \rightarrow \fP \fQ$, and $T: \fP \fQ \rightarrow \fQ \fP$.}

$\bullet$ The definition of $T:\fP \fP \rightarrow \fP \fP$. We use the isomorphisms
\begin{align}
P(n+1)\otimes_{A_{n+1}} P(n) &\cong  \left(A_{n+2}\otimes_{A_{n+1}\otimes A_1}(A_{n+1}\otimes{\C})\right)\otimes_{A_{n+1}}
  \left(A_{n+1}\otimes_{A_n\otimes A_1} (A_n\otimes {\C})\right) \notag \\
&\cong  A_{n+2}\otimes_{A_{n+1}\otimes A_1} (A_{n+1}\otimes {\C})\otimes_{A_n\otimes A_1} (A_n\otimes {\C}) \notag \\
&\cong  A_{n+2}\otimes_{A_n\otimes A_1\otimes A_1}(A_n\otimes {\C}\otimes {\C}) \label{e:PPtoPP1} \\
&\cong  A_{n+2}\otimes_{A_n\otimes A_2} A_n \otimes \left(A_2\otimes_{A_1\otimes A_1}({\C}\otimes {\C})\right)\label{e:PPtoPP2}
\end{align}
so that $P(n+1)\otimes_{A_{n+1}} P(n)$ can be identified with the expression (\ref{e:PPtoPP1})
or the expression (\ref{e:PPtoPP2}). We apply the map
\begin{equation*}
\begin{array}{rcl}
T: A_2\otimes_{A_1\otimes A_1}({\C}\otimes{\C}) & \longrightarrow  & A_2\otimes_{A_1\otimes A_1} ({\C}\otimes {\C})\\
           a\otimes 1 \otimes 1 & \longmapsto & as_1\otimes 1 \otimes 1, \qquad \text{for $a\in A_2, s_1=(12)\in S_2$},
\end{array}
\end{equation*}
where we notice that $a=\left((a_1\otimes a_2),\sigma\right)$ with $a_1,a_2\in A_1, \sigma\in S_2$ and
by definition $as_1=\left((a_1\otimes a_2),\sigma s_1\right)$. Therefore we get the induced map
\begin{equation*}
\begin{array}{rcl}
T: A_{n+2}\otimes_{A_n\otimes A_1\otimes A_1} (A_n\otimes {\C}\otimes {\C}) &\longrightarrow
  &A_{n+2}\otimes_{A_n\otimes A_1\otimes A_1} (A_n\otimes {\C}\otimes {\C})\\
1\otimes (1\otimes 1\otimes 1) &\longmapsto & s_{n+1}\otimes (1\otimes 1\otimes 1).
\end{array}
\end{equation*}
Thus we get the natural transformation $T: \fP \fP \rightarrow \fP \fP$.

$\bullet$ The definition of $T: \fQ\fQ \rightarrow \fQ\fQ$ is the same as above.

$\bullet$ The definition of $T: \fQ\fP \rightarrow \fP\fQ$. Here we use the homomorphism
\begin{equation*}
(A_n\otimes {\C})\otimes_{A_n\otimes A_1} A_{n+1}\otimes_{A_n\otimes A_1} (A_n\otimes {\C})
\longrightarrow A_n\otimes_{A_{n-1}\otimes A_1}(A_{n-1}\otimes {\C}\otimes{\C})\otimes_{A_{n-1}\otimes A_1} A_n
\end{equation*}
of $(A_n,A_n)$-bimodules determined by
\begin{equation*}
(1\otimes 1)\otimes a\otimes (1\otimes 1) \longmapsto
\begin{cases}
0, &\text{ if $a=1$;}\\
1\otimes (1\otimes 1\otimes 1)\otimes 1, &\text{ if $a=s_n=(n,n+1)$}.
\end{cases}
\end{equation*}
Then in the same way we get the mentioned natural transformation.

$\bullet$ The definition of $T: \fP\fQ \rightarrow \fQ\fP$. We use the homomorphism
\begin{equation*}
\begin{array}{rcl}
A_n\otimes_{A_{n-1}\otimes A_1}(A_{n-1}\otimes{\C}\otimes{\C})\otimes_{A_{n-1}\otimes A_1} A_n
&\longrightarrow &(A_n\otimes{\C})\otimes_{A_n\otimes A_1} A_{n+1}\otimes_{A_n\otimes A_1}(A_n\otimes{\C})\\
1\otimes (1\otimes 1\otimes 1)\otimes 1 &\longmapsto &(1\otimes 1)\otimes s_n\otimes (1\otimes 1)
\end{array}
\end{equation*}
to get the mentioned natural transformation.

{\bf (4). The definitions of $\adj: \fQ \fP \rightarrow \id$ and $\adj: \fP \fQ \rightarrow \id [1]\{-1\}$.}

$\bullet$ The definition of $\adj: \fQ \fP \rightarrow \id$. It suffices to define an appropriate
homomorphism of $(A_n,A_n)$-bimodules
\begin{equation*}
\adj: (n)Q\otimes_{A_{n+1}} P(n) \longrightarrow A_n.
\end{equation*}
Applying the map
\begin{equation*}
{\C}[S_{n+1}] \longrightarrow {\C}[S_n]
\end{equation*}
which sends $1$ to $1$ and $s_n=(n,n+1)$ to $0$, we obtain a map $A_{n+1}\rightarrow A_n\otimes A_1$ and
the induced map in the following is the required homomorphism
\begin{equation*}
\begin{array}{rl}
(n)Q\otimes_{A_{n+1}} P(n)
= &(A_n\otimes {\C})\otimes_{A_n\otimes A_1} A_{n+1}\otimes_{A_{n+1}}
  A_{n+1}\otimes_{A_n\otimes A_1} (A_n\otimes {\C})\\
\longrightarrow &(A_n\otimes{\C})\otimes_{A_n\otimes A_1} (A_n\otimes A_1)\otimes_{A_n\otimes A_1}(A_n\otimes {\C}) \cong A_n .
\end{array}
\end{equation*}

$\bullet$ The definition of $\adj: \fP \fQ \rightarrow \id [1]\{-1\}$. We need a map of $(A_{n+1} , A_{n+1} )$-bimodules
$$P (n) \otimes_{A_n } (n)Q  \rightarrow A_{n+1} [1]\{-1\}.$$
Since
\begin{eqnarray*}
P (n) \otimes_{A_n } (n)Q  = A_{n+1}  \otimes_{A_n  \otimes A_1 } (A_n  \otimes \C  \otimes \C ) \otimes_{A_n  \otimes A_1 } A_{n+1}
\end{eqnarray*}
It suffices to define an appropriate map $h: \C  \otimes \C  \rightarrow A_1  [1]\{-1\}$ of graded $(A_1 , A_1 )$-bimodules, so that we define $\adj$ as the composition of the following sequence of maps
$$
\adj: A_{n+1}  \otimes_{A_n  \otimes A_1 } (A_n  \otimes \C  \otimes \C ) \otimes_{A_n  \otimes A_1 } A_{n+1}  \xrightarrow{h} A_{n+1}  \otimes_{A_n  \otimes A_1 } A_{n+1}  [1]\{-1\} \rightarrow A_{n+1}  [1]\{-1\}
$$
where the second map is multiplication.

In the derived category, $\C$ can be represented by the complex
$0 \rightarrow  A_1  \otimes V^\vee \xrightarrow{d} A_1  \rightarrow 0$
or the complex
$0 \rightarrow   V^\vee \otimes A_1 \xrightarrow{d} A_1  \rightarrow 0$, due to the
resolutions (\ref{e:kres1}) and (\ref{e:kres2}), hence
$\C\otimes \C$ can be represented by the complex
$$\xymatrix{
0 \ar[r] &(A_1\otimes V^\vee)\otimes (V^\vee\otimes A_1) \ar[r] &((A_1\otimes V^\vee)\otimes A_1)\oplus (A_1\otimes (V^\vee \otimes A_1)) \ar[r] &A_1\otimes A_1\ar[r] &0 .
}$$
Then we define $h$ by the commutative diagram
\begin{equation*}
\xymatrix{
0 \ar[r] &(A_1\otimes V^\vee)\otimes (V^\vee\otimes A_1) \ar[r]\ar[d] &((A_1\otimes V^\vee)\otimes A_1)\oplus (A_1\otimes (V^\vee \otimes A_1)) \ar[r]\ar[d] &A_1\otimes A_1\ar[r]\ar[d] &0 \\
0 \ar[r] &0 \ar[r] &A_1\{-1\} \ar[r] &0 \ar[r] &0 .
}
\end{equation*}
where the map in the middle column is given as follows:
\begin{itemize}
\item[$\star$] $(A_1\otimes V^\vee)\otimes A_1\rightarrow A_1$ is defined by $(f_1\otimes \omega)\otimes f_2 \mapsto\la\omega,v\ra f_1f_2$,
\item[$\star$] $A_1\otimes (V^\vee \otimes A_1)\rightarrow A_1$ is defined by $g_1\otimes (\omega\otimes g_2)\mapsto -\la \omega,v\ra g_1g_2$.
\end{itemize}

{\bf (5). The definitions of $\adj: \id \rightarrow \fQ \fP [-1]\{1\}$ and $\adj: \id \rightarrow \fP \fQ$.}

$\bullet$ The definition of $\adj: \id \rightarrow \fQ \fP [-1]\{1\}$. It suffices to define a map of $(A_n ,A_n )$-bimodules
$$A_n  \longrightarrow (n)Q  \otimes_{A_{n+1} } P (n) [-1]\{1\}$$
which is equivalent to the map
\begin{equation}\label{e:idqp}
A_n  [1]\{-1\} \longrightarrow (A_n  \otimes \C ) \otimes_{A_n  \otimes A_1 } A_{n+1}  \otimes_{A_1  \otimes A_n } (\C  \otimes A_n ) .
\end{equation}
We replace the first $\C$ and the second $\C$ with their isomorphic complexes by the resolutions (\ref{e:kres1}) and (\ref{e:kres2}). Then the map (\ref{e:idqp}) comes from the following morphism of complexes in derived
category
$$
\begin{array}{crc}
  0 & \longrightarrow & 0 \\[0.5mm]
  \downarrow  &  & \downarrow \\[0.5mm]
  0 & \qquad\longrightarrow & (A_n\otimes(V^\vee\otimes A_1))\otimes_{A_n\otimes A_1}A_{n+1}\otimes_{A_n\otimes A_1}(A_n\otimes(A_1\otimes V^\vee)) \\[0.5mm]
  \downarrow  &  & \downarrow \\[0.5mm]
  A_n\{-1\} & \longrightarrow & (A_n\otimes(V^\vee\otimes A_1))\otimes_{A_n\otimes A_1}A_{n+1}\otimes_{A_n\otimes A_1}(A_n\otimes A_1)\\[0.5mm]
   & &  \hspace{90pt} \oplus (A_n\otimes A_1)\otimes_{A_n\otimes A_1}A_{n+1}\otimes_{A_n\otimes A_1}(A_n\otimes(A_1\otimes V^\vee)) \\[0.5mm]
   \downarrow  &  & \downarrow \\[0.5mm]
  0  & \longrightarrow & (A_n\otimes A_1)\otimes_{A_n\otimes A_1}A_{n+1}\otimes_{A_n\otimes A_1}(A_n\otimes A_1)=A_{n+1} \\[0.5mm]
  \downarrow  &  & \downarrow \\[0.5mm]
  0 &  \longrightarrow & 0
\end{array}
$$
where $A_n\{-1\} \longrightarrow (A_n\otimes(V^\vee\otimes A_1))\otimes_{A_n\otimes A_1}A_{n+1}\otimes_{A_n\otimes A_1}(A_n\otimes A_1)
\oplus (A_n\otimes A_1)\otimes_{A_n\otimes A_1}A_{n+1}\otimes_{A_n\otimes A_1}(A_n\otimes(A_1\otimes V^\vee))$ is uniquely determined by sending $1\longmapsto (1\otimes (v^\vee\otimes 1))\otimes 1\otimes (1\otimes 1)-(1\otimes 1)\otimes 1\otimes((1\otimes v^\vee)\otimes 1)$. Here we need to notice that the left column
has $A_n\{-1\}$ at the degree $-1$ position of the complex, thus represents the object $A_n[1]\{-1\}$ in
the derived category. It is also easy to check that the map is well defined.

$\bullet$ The definition of $\adj: \id \rightarrow \fP \fQ$. This functor is induced from the following
map of $(A_{n+1}, A_{n+1})$-bimodules
\begin{equation*}
\begin{array}{rcl}
A_{n+1} & \longrightarrow & P(n)\otimes_{A_n} (n)Q = A_{n+1}\otimes_{A_n\otimes A_1}(A_n\otimes{\C}\otimes{\C})
    \otimes_{A_n\otimes A_1} A_{n+1}\\
1       & \longmapsto & \sum\limits_{i=0}^n s_i\cdots s_n\otimes (1\otimes 1\otimes 1)\otimes s_n\cdots s_i.
\end{array}
\end{equation*}
Here we use the convention $s_0s_1\cdots s_n=1$.

\begin{theorem}[\cite{CL}]\label{theorem:2}
The natural transformations $X$, $T$ and $\adj$ satisfy the Heisenberg $2$-relations and give a categorical Heisenberg action of $'\H$ on $\oplus_{n \ge 0} D({\mathbf A_n\gmod})$.
\end{theorem}
\begin{proof}
The proof of this theorem is very similar to that of Theorem 2 in Section 4.4 of \cite{CL}. We will describe what needs to be proved and omit the detail.
The Heisenberg 2-relations are just the local relations mentioned in Section \ref{section:2}, they contain the following items, where we denote by $I$ the identity natural transformation on functors.

{\bf $\bullet$ Adjoint relations.} The following compositions involving adjunctions are all equal to the identity 2-morphism:
\begin{itemize}
\item[(1)] $\fP\xrightarrow{I\adj} \fP\fQ\fP[-1]\{1\}\xrightarrow{\adj I} \fP$ \quad and \quad
$\fQ\xrightarrow{\adj I}\fQ\fP\fQ[-1]\{1\}\xrightarrow{I\adj}\fQ$;
\item[(2)] $\fP\xrightarrow{\adj I}\fP\fQ\fP\xrightarrow{I\adj}\fP$ \quad and \quad
$\fQ\xrightarrow{I\adj}\fQ\fP\fQ\xrightarrow{\adj I} \fQ$.
\end{itemize}
In graphical calculus, the above relations are depicted as
\begin{equation*}
\begin{tikzpicture}[>=stealth]
\draw (0,0) -- (0,-0.7)[-];
\draw (0,0) arc(180:0:0.4)[-];
\draw (0.8,0) arc(180:360:0.4)[-];
\draw (1.6,0) -- (1.6,0.7)[-];
\draw [shift={(2.4,0)}] node{$=$};
\draw [shift={(3.2,0)}] (0,-0.7) -- (0,0.7)[-];
\draw [shift={(4,0)}] node{$=$};
\draw [shift={(4.8,0)}] (0,0) -- (0,0.7)[-];
\draw [shift={(4.8,0)}] (0,0) arc(180:360:0.4)[-];
\draw [shift={(5.6,0)}] (0,0) arc(180:0:0.4)[-];
\draw [shift={(6.4,0)}] (0,0) -- (0,-0.7)[-];
\end{tikzpicture} \qquad .
\end{equation*}
That arrows are not drawn in the graph means the above equalities are true for both orientations given
by (1) and (2).

{\bf $\bullet$ Dots and adjunctions.} For $\forall b \in {\Lambda}^*(V)$, we have the following equalities
of compositions of maps:
\begin{itemize}
\item[(1)] $(\fQ\fP\xrightarrow{IX(b)}\fQ\fP\xrightarrow{\adj}\id) = (\fQ\fP\xrightarrow{X(b)I}
\fQ\fP\xrightarrow{\adj}\id) $;
\item[(2)] $(\fP\fQ\xrightarrow{X(b)I}\fP\fQ\xrightarrow{\adj}\id[1]\{-1\}) = (\fP\fQ\xrightarrow{IX(b)}
\fP\fQ\xrightarrow{\adj}\id[1]\{-1\})$;
\item[(3)] $(\id\xrightarrow{\adj}\fP\fQ\xrightarrow{X(b)I}\fP\fQ) = (\id\xrightarrow{\adj}\fP\fQ
\xrightarrow{IX(b)}\fP\fQ)$;
\item[(4)] $(\id\xrightarrow{\adj}\fQ\fP[-1]\{1\}\xrightarrow{IX(b)}\fQ\fP[-1]\{1\}) =
(\id\xrightarrow{\adj}\fQ\fP[-1]\{1\}\xrightarrow{X(b)I}\fQ\fP[-1]\{1\})$.
\end{itemize}
The above equalities can be depicted by the following:
\begin{equation*}
\begin{tikzpicture}[>=stealth]
\draw (0,0) -- (0,-0.4)[->];
\draw (0,0) arc (180:0:0.4);
\draw [shift={(0.8,0)}] (0,0) -- (0,-0.4)[-];
\filldraw [shift={(0.8,0)}] (0,-0.2) circle(2pt) +(0.15,0) node {$b$};
\draw [shift={(1.3,0)}] node{$=$};
\draw [shift={(1.8,0)}] (0,0) -- (0,-0.4)[->];
\filldraw [shift={(1.8,0)}] (0,-0.2) circle(2pt) +(-0.15,0) node{$b$};
\draw [shift={(1.8,0)}] (0,0) arc(180:0:0.4);
\draw [shift={(2.6,0)}] (0,0) -- (0,-0.4)[-];
\draw [shift={(2.8,-0.3)}] node{$,$};

\draw [shift={(3.5,0)}] (0,0) -- (0,-0.4)[-];
\draw [shift={(3.5,0)}] (0,0) arc (180:0:0.4);
\draw [shift={(4.3,0)}] (0,0) -- (0,-0.4)[->];
\filldraw [shift={(3.5,0)}] (0,-0.2) circle(2pt) +(-0.15,0) node {$b$};
\draw [shift={(4.8,0)}] node{$=$};
\draw [shift={(5.3,0)}] (0,0) -- (0,-0.4)[-];
\filldraw [shift={(6.1,0)}] (0,-0.2) circle(2pt) +(0.15,0) node{$b$};
\draw [shift={(5.3,0)}] (0,0) arc(180:0:0.4);
\draw [shift={(6.1,0)}] (0,0) -- (0,-0.4)[->];
\draw [shift={(6.5,-0.3)}] node{$,$};

\draw [shift={(7,0)}] (0,0) -- (0,0.4)[->];
\draw [shift={(7,0)}] (0,0) arc (180:360:0.4);
\draw [shift={(7.8,0)}] (0,0) -- (0,0.4)[-];
\filldraw [shift={(7,0)}] (0,0.2) circle(2pt) +(-0.15,0) node {$b$};
\draw [shift={(8.3,0)}] node{$=$};
\draw [shift={(8.8,0)}] (0,0) -- (0,0.4)[->];
\draw [shift={(8.8,0)}] (0,0) arc(180:360:0.4);
\draw [shift={(9.6,0)}] (0,0) -- (0,0.4)[-];
\filldraw [shift={(9.6,0)}] (0,0.2) circle(2pt) +(0.15,0) node{$b$};
\draw [shift={(9.8,-0.3)}] node{$,$};

\draw [shift={(10.5,0)}] (0,0) -- (0,0.4)[-];
\draw [shift={(10.5,0)}] (0,0) arc (180:360:0.4);
\draw [shift={(11.3,0)}] (0,0) -- (0,0.4)[->];
\filldraw [shift={(11.3,0)}] (0,0.2) circle(2pt) +(0.15,0) node {$b$};
\draw [shift={(11.8,0)}] node{$=$};
\draw [shift={(12.3,0)}] (0,0) -- (0,0.4)[-];
\draw [shift={(12.3,0)}] (0,0) arc(180:360:0.4);
\draw [shift={(13.1,0)}] (0,0) -- (0,0.4)[->];
\filldraw [shift={(12.3,0)}] (0,0.2) circle(2pt) +(-0.15,0) node{$b$};
\draw [shift={(13.3,-0.3)}] node{$.$};
\end{tikzpicture}
\end{equation*}

{\bf $\bullet$ Pitchfork relations.} We have the following equalities for compositions of maps:
\begin{itemize}
\item[(1)] $(\fP\xrightarrow{I\adj}\fP\fP\fQ\xrightarrow{TI}\fP\fP\fQ)
=(\fP\xrightarrow{\adj I}\fP\fQ\fP\xrightarrow{IT}\fP\fP\fQ)$;
\item[(2)] $(\fQ\xrightarrow{I\adj}\fQ\fP\fQ\xrightarrow{TI}\fP\fQ\fQ)
=(\fQ\xrightarrow{\adj I}\fP\fQ\fQ\xrightarrow{IT}\fP\fQ\fQ)$;
\item[(3)] $(\fP\fP\fQ\xrightarrow{IT}\fP\fQ\fP\xrightarrow{\adj I}\fP[1]\{-1\})
=(\fP\fP\fQ\xrightarrow{TI}\fP\fP\fQ\xrightarrow{I\adj}\fP[1]\{-1\})$;
\item[(4)] $(\fP\fQ\fQ\xrightarrow{IT}\fP\fQ\fQ\xrightarrow{\adj I}\fQ[1]\{-1\})
=(\fP\fQ\fQ\xrightarrow{TI}\fQ\fP\fQ\xrightarrow{I\adj}\fQ[1]\{-1\})$.
\end{itemize}
The above equalities are actually rel boundary isotopy relations for 2-morphisms:
\begin{equation*}
\begin{tikzpicture}[>=stealth]
\draw (0,0) -- (0,0.2)[->];
\draw (0,0) arc (180:360:0.4);
\draw [shift={(0.8,0)}] (0,0) -- (0,0.2)[-];
\draw (0,-0.7) -- (0.4,0.2)[->];
\draw [shift={(1.3,0)}] node{$=$};
\draw [shift={(1.8,0)}] (0,0) -- (0,0.2)[->];
\draw [shift={(1.8,0)}] (0,0) arc(180:360:0.4);
\draw [shift={(2.6,0)}] (0,0) -- (0,0.2)[-];
\draw [shift={(2.6,0)}] (0,-0.7) -- (-0.4,0.2)[->];
\draw [shift={(2.8,-0.3)}] node{$,$};

\draw [shift={(3.5,0)}] (0,0) -- (0,0.2)[->];
\draw [shift={(3.5,0)}] (0,0) arc (180:360:0.4);
\draw [shift={(4.3,0)}] (0,0) -- (0,0.2)[-];
\draw [shift={(3.5,0)}] (0,-0.7) -- (0.4,0.2)[<-];
\draw [shift={(4.8,0)}] node{$=$};
\draw [shift={(5.3,0)}] (0,0) -- (0,0.2)[->];
\draw [shift={(5.3,0)}] (0,0) arc(180:360:0.4);
\draw [shift={(6.1,0)}] (0,0) -- (0,0.2)[-];
\draw [shift={(6.1,0)}] (0,-0.7) -- (-0.4,0.2)[<-];
\draw [shift={(6.5,-0.3)}] node{$,$};

\draw [shift={(7,0)}] (0,-0.2) -- (0,-0.4)[-];
\draw [shift={(7,-0.2)}] (0,0) arc (180:0:0.4);
\draw [shift={(7.8,0)}] (0,-0.2) -- (0,-0.4)[->];
\draw [shift={(7.8,0)}] (-0.4,-0.4) -- (0,0.4)[->];
\draw [shift={(8.3,0)}] node{$=$};
\draw [shift={(8.8,0)}] (0,-0.2) -- (0,-0.4)[-];
\draw [shift={(8.8,-0.2)}] (0,0) arc(180:0:0.4);
\draw [shift={(9.6,0)}] (0,-0.2) -- (0,-0.4)[->];
\draw [shift={(9.6,0)}] (-0.4,-0.4) -- (-0.8,0.4)[->];
\draw [shift={(9.8,-0.3)}] node{$,$};

\draw [shift={(10.5,0)}] (0,-0.2) -- (0,-0.4)[-];
\draw [shift={(10.5,-0.2)}] (0,0) arc (180:0:0.4);
\draw [shift={(11.3,0)}] (0,-0.2) -- (0,-0.4)[->];
\draw [shift={(11.3,0)}] (-0.4,-0.4) -- (0,0.4)[<-];
\draw [shift={(11.8,0)}] node{$=$};
\draw [shift={(12.3,0)}] (0,-0.2) -- (0,-0.4)[-];
\draw [shift={(12.3,-0.2)}] (0,0) arc(180:0:0.4);
\draw [shift={(13.1,0)}] (0,-0.2) -- (0,-0.4)[->];
\draw [shift={(13.1,0)}] (-0.4,-0.4) -- (-0.8,0.4)[<-];
\draw [shift={(13.3,-0.3)}] node{$.$};
\end{tikzpicture}
\end{equation*}

{\bf $\bullet$ Dots and crossings.} The following compositions of maps with graphs containing dots are
equal for $\forall b\in {\Lambda}^*(V)$:
\begin{itemize}
\item[(1)] $(\fP\fP\xrightarrow{X(b) I}\fP\fP\xrightarrow{T}\fP\fP)
=(\fP\fP\xrightarrow{T}\fP\fP\xrightarrow{I X(b)}\fP\fP)$;
\item[(2)] $(\fP\fP\xrightarrow{T}\fP\fP\xrightarrow{X(b) I}\fP\fP)
=(\fP\fP\xrightarrow{I X(b)}\fP\fP\xrightarrow{T}\fP\fP)$.
\end{itemize}
These equalities correspond to the following graphical relations:
\begin{equation*}
\begin{tikzpicture}[>=stealth, baseline=25pt]
\draw [shift={+(0,0)}](0,0) -- (1,1)[<-];
\draw [shift={+(0,0)}](1,0) -- (0,1)[->];
\filldraw [shift={+(0,0)}](0.3,0.3) circle (2pt)+(0.45,0.4) node {$b$};
\draw [shift={+(1.85,0.5)}] node{ $=$};
\draw [shift={+(2.7,0)}](0,0) -- (1,1)[<-];
\draw [shift={+(2.7,0)}](1,0) -- (0,1)[->];
\filldraw [shift={+(2.7,0)}](0.8,0.8) circle (2pt)+(0.95,0.75) node {$b$};
\draw [shift={+(3.7,0)}] (0.3,0) node{,};
\draw [shift={+(6,0)}](0,0) -- (1,1)[<-];
\draw [shift={+(6,0)}](1,0) -- (0,1)[->];
\filldraw [shift={+(6,0)}](0.2,0.8) circle (2pt)+(0.35,0.9) node {$b$};
\draw [shift={+(7.85,0.5)}] node{ $=$};
\draw [shift={+(8.7,0)}](0,0) -- (1,1)[<-];
\draw [shift={+(8.7,0)}](1,0) -- (0,1)[->];
\filldraw [shift={+(8.7,0)}](0.8,0.2) circle (2pt)+(0.95,0.3) node {$b$};
\draw [shift={+(9.7,0)}] (0.3,0) node{.};
\end{tikzpicture}
\end{equation*}

{\bf $\bullet$ Compositions of crossings relations.} The following equalities hold for compositions of
crossings relations:
\begin{itemize}
\item[(1)] $(\fP\fP\xrightarrow{T}\fP\fP\xrightarrow{T}\fP\fP)
=(\fP\fP\xrightarrow{II}\fP\fP)$;
\item[(2)] $(\fP\fP\fP\xrightarrow{IT}\fP\fP\fP\xrightarrow{TI}\fP\fP\fP\xrightarrow{IT}\fP\fP\fP)
=(\fP\fP\fP\xrightarrow{TI}\fP\fP\fP\xrightarrow{IT}\fP\fP\fP\xrightarrow{TI}\fP\fP\fP)$;
\item[(3)] $(\fP\fQ\xrightarrow{T}\fQ\fP\xrightarrow{T}\fP\fQ)=(\fP\fQ\xrightarrow{II}\fP\fQ)$;
\item[(4)] $(\fQ\fP\xrightarrow{II}\fQ\fP)=
(\fQ\fP\xrightarrow{T}\fP\fQ\xrightarrow{T}\fQ\fP)+
(\fQ\fP\xrightarrow{\adj}\id\xrightarrow{\adj}\fQ\fP[-1]\{1\}\xrightarrow{IX(v)}\fQ\fP)$\linebreak $
+\,\,(\fQ\fP\xrightarrow{IX(v)}\fQ\fP[1]\{-1\}\xrightarrow{\adj}\id[1]\{-1\}
\xrightarrow{\adj}\fQ\fP)$.
\end{itemize}
These equalities correspond to the following graphical relations:
\begin{equation*}
\begin{tikzpicture}[>=stealth]
\draw (0,0) .. controls (0.7,0.7) .. (0,1.4)[->];
\draw (0.7,0) .. controls (0,0.7) .. (0.7,1.4)[->] ;
\draw (1.2,0.7) node {=};
\draw (1.9,0) --(1.9,1.4)[->];
\draw (2.6,0) -- (2.6,1.4)[->];
\draw (2.6,0) node{\qquad ,};
\draw [shift={+(5.4,0)}](0,0) -- (1.4,1.4)[->];
\draw [shift={+(5.4,0)}](1.4,0) -- (0,1.4)[->];
\draw [shift={+(5.4,0)}](0.7,0) .. controls (0,0.7) .. (0.7,1.4)[->];
\draw [shift={+(5.7,0)}](1.7,0.7) node {=};
\draw [shift={+(6.0,0)}](2.1,0) -- (3.5,1.4)[->];
\draw [shift={+(6,0)}](3.5,0) -- (2.1,1.4)[->];
\draw [shift={(6,0)}] (3.5,0) node{\qquad ,};
\draw [shift={+(6,0)}](2.8,0) .. controls (3.5,0.7) .. (2.8,1.4)[->];
\end{tikzpicture}
\end{equation*}
\begin{equation*}
\begin{tikzpicture}[>=stealth]
\draw (0,0) .. controls (0.7,0.7) .. (0,1.4)[->];
\draw (0.7,0) .. controls (0,0.7) .. (0.7,1.4)[<-] ;
\draw (1.05,0.7) node {=};
\draw (1.61,0) --(1.61,1.4)[->];
\draw (2.31,0) node{\qquad ,} -- (2.31,1.4)[<-];

\draw [shift={(4.71,0)}](0,0) .. controls (0.7,0.7) .. (0,1.4)[<-];
\draw [shift={(4.71,0)}](0.7,0) .. controls (0,0.7) .. (0.7,1.4)[->];
\draw [shift={(4.71,0)}](1.05,0.7) node {=};
\draw [shift={(4.71,0)}](1.4,0) --(1.4,1.4)[<-];
\draw [shift={(4.71,0)}](2.1,0) -- (2.1,1.4)[->];
\draw [shift={(4.71,0)}](2.45,0.7) node {$-$};
\draw [shift={(4.71,0)}](2.8,1.225) arc (180:360:.35);
\draw [shift={(4.71,0)}](2.8,1.4) -- (2.8,1.225) ;
\draw [shift={(4.71,0)}](3.5,1.4) -- (3.5,1.225) [<-];
\draw [shift={(4.71,0)}](3.5,.175) arc (0:180:.35) ;
\filldraw  [shift={(4.71,0)}](3.15,0.875) circle (2pt)+(.175,0) node {$v$};
\draw [shift={(4.71,0)}](3.5,0) -- (3.5,.175) ;
\draw [shift={(4.71,0)}](2.8,0) -- (2.8,.175) [<-];
\draw [shift={(4.71,0)}](3.85,0.7) node {$-$};
\filldraw  [shift={(4.71,0)}](4.55,0.525) circle (2pt)+(.175,0) node {$v$};
\draw [shift={(4.71,0)}](4.2,1.225) arc (180:360:.35);
\draw [shift={(4.71,0)}](4.2,1.4) -- (4.2,1.225) ;
\draw [shift={(4.71,0)}](4.9,1.4) -- (4.9,1.225) [<-];
\draw [shift={(4.71,0)}](4.9,.175) arc (0:180:.35);
\draw [shift={(4.71,0)}](4.9,0) -- (4.9,.175) node{\qquad .} ;
\draw [shift={(4.71,0)}](4.2,0) -- (4.2,.175) [<-];
\end{tikzpicture}
\end{equation*}

{\bf $\bullet$ Counter-clockwise circles and curls.} The following equalities hold for compositions of
maps:
\begin{itemize}
\item[(1)] $(\fP\xrightarrow{\adj I}\fQ\fP\fP[-1]\{1\}\xrightarrow{IT}
\fQ\fP\fP[-1]\{1\}\xrightarrow{\adj I}\fP[-1]\{1\})=0$;
\item[(2)] $(\id\xrightarrow{\adj}\fQ\fP[-1]\{1\}\xrightarrow{IX(b)}\fQ\fP[-1+|b|]\{1-|b|\}\xrightarrow{\adj}
\id[-1+|b|]\{1-|b|\}) = (\id\xrightarrow{{\tr}(b)I}\id)$ for $\forall b\in {\Lambda}^*(V)$.
\end{itemize}
These equalities correspond to the following graphical relations:
\begin{equation*}
\begin{tikzpicture}[>=stealth]
\draw (-0.8,0) .. controls (-0.8,0.4) and (-.24,.4) .. (-.08,0) ;
\draw (-0.8,0) .. controls (-0.8,-.4) and (-.24,-.4) .. (-.08,0) ;
\draw (0,-0.8) .. controls (0,-.4) .. (-.08,0) ;
\draw (-.08,0) .. controls (0,.4) .. (0,0.8) [->] ;
\draw [shift={(0.4,0)}](0.50,0) node {$=\,\, 0$\, ,};

\draw [shift={+(4,0)}](0,0) arc (180:360:0.4cm) ;
\draw [shift={+(4,0)}][->](0.8,0) arc (0:180:0.4cm) ;
\filldraw [shift={+(4.8,0)}](0,0) circle (1.6pt);
\draw [shift={+(4,0)}](0.8,0) node [anchor=east] {$b$};
\draw [shift={+(4,0)}](1.8,0) node{$ = \,\,\tr(b)$\,.};
\end{tikzpicture}
\end{equation*}
\end{proof}

\begin{Remark}\label{rmk:2}
In the definition (\ref{e:q}) of the functor $\fQ$, we have
$(n)\fQ: D({\mathbf A_{n+1}  \gmod}) \rightarrow D({\mathbf A_n  \gmod})$ by
$
(n)\fQ(\cdot) := (n)Q  \otimes_{A_{n+1} } (\cdot).
$
If we replace the above by a formal shifting
\[
(n)\fQ(\cdot) := (n)Q  \otimes_{A_{n+1} } (\cdot)[-1/2]\{1/2\},
\]
then the construction of this section goes through with a little
change. Therefore we can get another categorified version of some
deformed Heisenberg algebra, with the basic isomorphism (\ref{e:basic})
replaced by
\[
QP\cong PQ \oplus \mathbf{1}\lba -1/2\rba\oplus \mathbf 1 \lba 1/2\rba.
\]
The shifted degree $sdeg$ in Section \ref{section:2} corresponds to
this construction.
\end{Remark}

\begin{Remark}\label{rmk:2.5}
Let ${\mathfrak F}$ be the $2$-category with only one $0$-cell equal to $\bigoplus_{n\ge 0}D({\mathbf A_n}{\gmod})$,
the $1$-cells equal to the endofunctors of $\bigoplus_{n\ge 0}D({\mathbf A_n}\gmod)$, and
the $2$-cells equal to the natural transformations between these endofunctors.
Then there is an obvious $2$-functor from $'\mathcal H$ to $\mathcal F$, therein
$\bigoplus_{n\ge 0}D({\mathbf A_n}{\gmod})$
is regarded as the Fock space representation of $'\H$.
\end{Remark}


\section{The Grothendieck ring $K_0('{\mathcal H})$}\label{section:4}
In this section, we view  $'\H'$ and $'\H$ as the
categories by forgetting the $0$-cells $\Z$, hence the objects of $'\H'$ are generated
by $P$ and $Q$ through degree shiftings, finite compositions, and finite direct sums,
and the morphisms constitute the $\C$-vector space generated by the planar strings
described in Section \ref{section:2}. The category $'\H$ is the Karoubi envelope of
$'\H'$. By definition, $K_0('{\mathcal H})$ is the ring of isomorphism classes of
the 1-morphisms of $'\mathcal H$ with addition induced by the direct sum of $'\H$ and
multiplication induced by the composition of $'\H$.
We study $K_0('\H)$ by considering
the endomorphism ring ${\rm End}(M)_0$ of each $1$-morphism $M$ of $'\H'$:
\[
{\rm End}(M)_0=\{\alpha:M\to M \,\,|\,deg(\alpha)=0\},
\]
so that each $1$-morphism of $'\H$ can be expressed as $(M,f)$ with
$f\in {\rm End}(M)_0$.
The reason of using the endomorphism ring of degree $0$ instead of any degrees is that
we distinguish the dimension shifted $1$-morphisms in $'\H'$. For example,
$P$ and $P\lba 1\rba$ are not considered isomorphic. In fact, we have
\[
[P\lba 1\rba]=[P]t.
\]

By the basic isomorphism (\ref{e:basic}), any $1$-morphism of $'\H'$ is isomorphic to a finite direct
sum of the form
\begin{equation}\label{e:ob1}
M_{{\m},{\n},{\l}}=\bigoplus_{i=1}^k P^{m_i}Q^{n_i}\lba l_i\rba,
\end{equation}
where we denote by ${\m}=(m_1,\cdots,m_k), {\n}=(n_1,\cdots,n_k),$ and ${\l}=(l_1,\cdots,l_k)$
as tuples of integers, and every $m_i, n_i\ge 0$. Let $R_{\m,\n,\l}={\rm End}(M_{\m,\n,\l})_0$.
Then an element $A\in R_{\m,\n,\l}$ can be written as
\[
A=\left(
\begin{array}{cccc}
a_{1,1} & a_{1,2} & \cdots & a_{1,k} \\
a_{2,1} & a_{2,2} & \cdots & a_{2,k} \\
\cdots  & \cdots  & \cdots & \cdots \\
a_{k,1} & a_{k,2} & \cdots & a_{k,k}
\end{array}
\right)
\]
with the entries $a_{i,j}\in {\rm Hom}(P^{m_i}Q^{n_i}\lba l_i\rba, P^{m_j}Q^{n_j}\lba l_j\rba)_0$.
View each $a_{i,j}$ as an element of \linebreak ${\rm Hom}$$(P^{m_i}Q^{n_i}, $$P^{m_j}$$Q^{n_j})$,
we know from Proposition \ref{prop:1} that $sdeg(a_{i,j})\ge 1/2$ unless
$(m_i,n_i)=(m_j,n_j)$. We get a filtration on $R_{\m,\n,\l}$:
\begin{equation}\label{e:filtration1}
R_{\m,\n,\l}=F_0\supsetneq F_1\supset F_2\supset \cdots \supset F_p\supset F_{p+1}\supset\cdots
\end{equation}
satisfying $F_p\cdot F_q\subset F_{p+q}$ for $p,q\ge 0$ and $\bigcap_{p=0}^{\infty}F_p=0$, where
\[
F_p=\{A=(a_{i,j}\in R_{\m,\n,\l}\,\,|\, sdeg(a_{i,j})\ge p/2 \text{ for $1\le i,j\le k$}\}.
\]
We note that the filtration (\ref{e:filtration1}) is actually finite. We also denote by
\[
R^0_{\m,\n,\l}=\{A=(a_{i,j})\in R_{\m,\n,\l}\,\,|\,
\text{ $a_{i,j}$ is homogeneous of } sdeg(a_{i,j})=0 \text{ for $1\le i,j\le k$}\}.
\]
Note that the shifted degree $sdeg(a_{i,j})$ is for $a_{i,j}$ viewed as a morphism from
$P^{m_i}Q^{n_i}\to P^{m_j}Q^{n_j}$, hence $R_{\m,\n,\l}$ and $R^0_{\m,\n,\l}$ are different.
Moreover, $F_1$ is a two-sided ideal of $R_{\m,\n,\l}$, and we get a split exact sequence
\begin{equation}\label{e:exact1}
0\longrightarrow F_1\longrightarrow R_{\m,\n,\l}\longrightarrow R^0_{\m,\n,\l}\longrightarrow 0.
\end{equation}
We notice the $1-1$ correspondences
\[
\begin{array}{rcl}
\left\{\begin{array}{c}
       \text{$(M_{\m,\n,\l},e)$: $e\in R_{\m,\n,\l}$}\\
       \text{is an idempotent}
       \end{array}
\right\}
&\longleftrightarrow &
\left\{\begin{array}{c}
       \text{$eR_{\m,\n,\l}$: $e\in R_{\m,\n,\l}$}\\
       \text{is an idempotent}
       \end{array}
\right\}\\
(M_{\m,\n,\l},e) &\longmapsto & e\cdot R_{\m,\n,\l}
\end{array}
\]
and
\[
\begin{array}{rcl}
\left\{\begin{array}{l}
       \text{$e\in M_s(R_{\m,\n,\l})$:\,\, $e$ is an}\\
       \text{idempotent $s\times s$ matrix over}\\
       \text{$R_{\m,\n,\l}$ for some integer $s>0$}
       \end{array}
\right\}
&\longleftrightarrow &
\left\{\begin{array}{l}
       \text{finitely generated projective right}\\
       \text{$R_{\m,\n,\l}$-module $N$:\,\,$N$ is a direct}\\
       \text{summand of $R_{\m,\n,\l}^s$ for some}\\
       \text{integer $s>0$}
       \end{array}
\right\} \\
e &\longmapsto & e\cdot R_{\m,\n,\l}^s,
\end{array}
\]
and
\[
K_0(R_{\m,\n,\l})=
\left\{\begin{array}{l}
       \text{isomorphism classes of finitely generated projective }\\
       \text{right $R_{\m,\n,\l}$-module $N$:\,\,$N$ is a direct summand }\\
       \text{of $R_{\m,\n,\l}^s$ for some }
       \text{integer $s>0$}
       \end{array}
\right\}
\]
as sets.

\begin{Remark}
In the above discussion, we use the finitely generated projective right modules of rings to
define the Grothendieck rings instead of the left modules. But similar to what is pointed out by
Khovanov in \cite{K}, the category $'\H'$ has the symmetry given on the diagrams by reflecting
about the $x$-axis and reversing orientation, which is an involutive monoidal contravariant auto-equivalence
of $'\H'$. Such an equivalence implies that we get the same Grothendieck ring of $'\H$ for
considering either the left modules or the right modules.
\end{Remark}

\begin{lemma}\label{lemma:1}
Let $R$ be a unital ring containing $\Z$ with a filtration
\[
R=F_0\supsetneq F_1\supset F_2\supset \cdots\supset F_p\supset F_{p+1}\supset\cdots
\]
such that $F_p\cdot F_q\subset F_{p+q}$ for integers $p,q\ge 0$
(therefore $F_1$ is a two-sided ideal of $R$) and $\bigcap_{p=0}^{\infty}F_p=0$.
Suppose that $R_0$ is a unital subring
of $R$ such that $R_0\cong R/F_1$ and the exact sequence
\begin{equation}\label{e:lem11}
0\longrightarrow F_1\longrightarrow R\longrightarrow R_0\longrightarrow 0
\end{equation}
is split. Then $K_0(R)\cong K_0(R_0)$ (as abelian groups).
\end{lemma}
\begin{proof}
From the split exact sequence (\ref{e:lem11}), we get the exact sequence of Grothendieck groups
\begin{equation}\label{e:lem12}
0\lra K_0(F_1)\lra K_0(R)\lra K_0(R_0)\lra 0.
\end{equation}
$F_1$ is a non-unital ring with an embedding $j:F_1\longrightarrow F_1^+$, where
$F_1^+=F_1\oplus {\Z}$ as abelian groups and the multiplication is defined by
\[
(x,m)\cdot (y,n)=(xy + my + nx, mn) \text{ for $x,y\in F_1$ and $m,n\in{\Z}$}.
\]
For convenience, we write the element $(x,m)\in F_1^+$ by $x+m$. Then we have the
split exact sequence
\[
0\xrightarrow{} F_1\xrightarrow{j} F_1^+ \xrightarrow{\rho} {\Z}\xrightarrow{} 0.
\]
Then $K_0(F_1)={\rm ker}({\rho}_*:\, K_0(F_1^+)\longrightarrow K_0({\Z}))$. Since $F_1^+$ is a
unital ring, an element of $K_0(F_1^+)$ is some conjugacy class of an idempotent
square matrix over $F_1^+$ (see \cite{Ros} for more detail).
$F_1^+$ can be viewed as a subring of $R$, thus we have the filtration for $F_1^+$:
\[
F_1^+\supsetneq F_1\supset F_2\supset\cdots\supset F_p\supset F_{p+1}\supset\cdots.
\]

Suppose $0\ne e \in M_t(F_1^+)$ is an idempotent matrix of size
$t\times t$ such that ${\rho}_*(e)=0$. Let $N_p=M_t(F_p)$, we then get a
filtration on $M_t(F_1^+)$:
\[
N_0:=M_t(F_1^+)\supsetneq N_1\supset N_2\supset\cdots\supset N_p\supset N_{p+1}\supset\cdots
\]
satisfying $N_p\cdot N_q\subset N_{p+q}$.
Write $e=e_0+x$ with $e_0\in M_t(\Z)$ and $x\in N_1$. Since $e$ is an idempotent,
we have
\[
e_0^2+e_0 x + xe_0 +x^2 = e_0 + x.
\]
Since $N_p\cdot N_q\subset N_{p+q}$, we get
\[
e_0^2 = e_0, \qquad \text{ and } \qquad e_0 x + xe_0+x^2=x.
\]
As ${\rho}_*(e)=0$, we have $e_0=0$, thus $x=x^2$ and $x\ne 0$. Hence
there exists a positive integer $l$ such that $x\in N_l\backslash N_{l+1}$,
which contradicts to $x=x^2\in N_{2l}\subset N_{l+1}$. Therefore we
have proved that $K_0(F_1)={\rm ker}(\rho)=0$. The lemma then follows from
the exact sequence (\ref{e:lem12}).
\end{proof}

\begin{cor}\label{cor:1}
Let $R_{\m,\n,\l}$ and $R_{\m,\n,\l}^0$ be as above, and
\[
X=\{(m_i,n_i,l_i)\,\,|\, \m=(m_1,m_2,\cdots, m_k), \n=(n_1,n_2,\cdots, n_k),
\l=(l_1,l_2,\cdots, l_k), 1\le i\le k\}
\]
the set of different $3$-tuples of $(m_i,n_i,l_i)$ appearing in the decomposition
(\ref{e:ob1}) of $M_{\m,\n,\l}$. Then
\[
K_0(R_{\m,\n,\l})\cong K_0(R_{\m,\n,\l}^0)\cong\bigoplus_{(m,n,l)\in X}
K_0({\rm End}(P^mQ^n)_0)\cong \bigoplus_{(m,n,l)\in X}K_0({\C}[S_m]\otimes {\C}[S_n])
\]
as abelian groups. Therefore, an element $[(M_{\m,\n,\l}\,,e)]$
(where $e\in R_{\m,\n,\l}$ is an idempotent) of $K_0('\H)$ can be expressed
as
\[
\sum_{(m,n,l)\in X}\,\,\sum_{\lambda,\mu}c_{m,n,l,\lambda,\mu}
[(P^m,e_{\lambda})]\cdot [(Q^n, e_{\mu})]\cdot t^l
\]
for non-negative integers $c_{m,n,l,\lambda,\mu}$, where $\lambda$ is a
partition of $m$ and $\mu$ is a partition of $n$.
\end{cor}
\begin{proof}
The first isomorphism is direct from Lemma \ref{lemma:1}, the second isomorphism
is from the descriptions of $2$-morphisms of shifted degree $0$ given in Proposition
\ref{prop:1} and the ``Morita invariance" theorem saying that $K_0(R)\cong K_0(M_n(R))$
for any unital ring $R$ and any positive integer $n$
(see for example, Chapter 1 of \cite{Ros}), the third
isomorphism holds because ${\rm End}(P^mQ^n)_0\cong {\C}[S_m]\otimes {\C}[S_n]$
as stated in Proposition \ref{prop:1}.
\end{proof}

\noindent{\bf Corollary \,\ref{cor:1}$'$.} {\it Let $\m$, $\n$, $\l$ be as in
Corollary \ref{cor:1}, and $M_{\m,\n,\l}$ be given in (\ref{e:ob1}). Let
\[
Y=\{(m_i,n_i)\,\,|\,\,\,{\m}=(m_i,n_i,l_i), {\n}=(n_1,n_2,\cdots,n_k), 1\le i\le k\}
\]
be the different pairs $(m_i,n_i)$ appearing in the decomposition (\ref{e:ob1}) of
$M_{\m,\n,\l}$. Then
\[
K_0({\End}(M_{\m,\n,\l}))\cong \bigoplus_{(m,n)\in Y}
K_0({\rm End}(P^mQ^n)_0)\cong \bigoplus_{(m,n)\in Y}K_0({\C}[S_m]\otimes {\C}[S_n]).
\]
}
\begin{proof}
Since every element of ${\End}(M_{\m,\n,\l}))$ is obtained from planar string diagrams,
we view elements of ${\End}(M_{\m,\n,\l}))$ as elements of ${\End}(M_{\m,\n,{\mathbf 0}}))$,
where ${\mathbf 0}=(0,0,\cdots, 0)$, thus get a similar filtration on ${\End}(M_{\m,\n,\l}))$
as in (\ref{e:filtration1}) and the conclusion follows.
\end{proof}

\begin{prop}\label{prop:2}
The Grothendieck ring $K_0('\mathcal H)$ is generated by the classes
$\{[P_m], [Q_n]:\,\, m,n\in{\Z}\text{ non-negative}\}$ as an algebra
over $\Z[t,t^{-1}]$.
\end{prop}
\begin{proof}
Let $[(M_{\m,\n,\l}, e)]$ be any element of $K_0('\H)$, where we keep the
notation $M_{\m,\n,\l}$ as in (\ref{e:ob1}). Corollary \ref{cor:1} implies
that $[(M_{\m,\n,\l}, e)]$ can be expressed in terms of $[(P^m, e_{\lambda})]\cdot
[(Q^n,e_{\mu})]$ with coefficients in ${\Z}[t,t^{-1}]$. Considering the
injection $S_{m_1}\times S_{m_2}\times \cdots \times S_{m_k}\hookrightarrow S_m$
with $m_1+m_2+\cdots +m_k=m$ for $m_i\ge 1$, we see that $[(P^m, e_{\lambda})]$
can be expressed in terms of $[P_i]=[(P^i, e_{(i)})]$ for non-negative integers $i$
over $\Z$ (see \cite{FH} or \cite{M}). Similar conclusion also holds for $[(Q^n, e_{\mu})]$. Applying Theorem \ref{theorem:1},
we get the proposition.
\end{proof}

Now we outline some results on the Fock space representation of the
Heisenberg algebra $'H_{\Zt}$ (given by the defining equation (\ref{e:arelation}))
in \cite{FJW1} and \cite{FJW2}, which is
needed in the proof of Theorem \ref{theorem:1.5}.

Let ${\mathfrak S}_{\Ccrs}$ be the symmetric algebra generated by $a_{-n}$ for positive integers $n$
over the ring $\C[t,t^{-1}]$. So ${\mathfrak S}_{\Ccrs}$ is a commutative subalgebra of the Heisenberg algebra
$H$ (also over $\C[t,t^{-1}]$), ${\Z}_+$-graded with $deg(a_{-n})=n$. An element
$a_n, n\ge 0$ acts on ${\mathfrak S}_{\Ccrs}$ as a derivation of algebra by
\begin{equation*}
a_n\cdot (a_{-n_1}a_{-n_2}\cdots a_{-n_k})
=\sum_{i=1}^k\delta_{n,n_i}n(1+t^n)\,a_{-n_1}a_{-n_2}\cdots \hat{a}_{-n_i}\cdots a_{-n_k}
\end{equation*}
according to the relation (\ref{e:arelation}).
Therefore we get an injection:
\[
'H_{\C[t,t^{-1}]}\hookrightarrow {\rm End}({\mathfrak S}_{\Ccrs})
\]
with $a_{-n}$ for $n\in {\Z}_+$ and $t^l$ for $l\in \Z$ acting on
${\mathfrak S}_{\Ccrs}$ by multiplication. The above injection
gives the Fock space representation ${\mathfrak S}_{\Ccrs}$ of $'H_{\Ct}$. We also set
\[
a_m(t^k)=a_m t^{mk} \text{ \,\, for $m\in \Z$.}
\]

For $k\in\Z$, let $t^k$ be the character of the $1$-dimensional representation ${\C}(k)$
of $\Ccrs$. Denote by
\begin{equation*}
\mathfrak R_{\Ccrs}=\bigoplus_{n\ge 0} {\mathfrak R}_{n,\Ccrs}, \qquad\qquad
{\mathfrak R}_{n,\Ccrs}=\bigoplus_i {\C}\gamma_i\otimes t^{k_i},
\end{equation*}
where $k_i\in \Z$, and $\gamma_i$ is an
irreducible character of the permutation group $S_n$ for each $i$.

The conjugacy classes of $S_n$ are given by partitions
$\lambda=(1^{m_1}2^{m_2}\cdots)$ corresponding to permutation types of
$S_n$ factorized as products of disjoint $m_1$ $1$-cycles, $m_2$ $2$-cycles,
$\cdots$. Write $\lambda=(\lambda_1,\lambda_2,\cdots)$ with
$\lambda_1\ge \lambda_2\ge \cdots$, and define
\[
a_{-\lambda}=a_{-\lambda_1}a_{-\lambda_2}\cdots, \quad\text{ and }\quad
a_{-\lambda\otimes t^k}=t^{-k|\lambda|} a_{-\lambda},
\]
where $|\lambda|=\lambda_1 +\lambda_2 +\cdots$ is the weight of $\lambda$.
The elements $a_{-\lambda}$, $\lambda$ taken over all partitions of positive integers
constitute a $\C[t,t^{-1}]$-basis of ${\mathfrak S}_{\Ccrs}$. Both ${\mathfrak R}_{\Ccrs}$
and ${\mathfrak S}_{\Ccrs}$ have bilinear form
structures and Hopf algebra structures. We also have the characteristic
map $ch: {\mathfrak R}_{\Ccrs}\lra {\mathfrak S}_{\Ccrs}$ defined by
\[
ch(f)=\frac{1}{n!}\sum_{\sigma\in S_n} S(f(\sigma)) a_{-\lambda(\sigma)}
\]
for $f\in {\mathfrak R}_{n,\Ccrs}$, where $\lambda(\sigma)$ denotes the partition of $n$
corresponding to the type of the permutation $\sigma$,
\[
f(\sigma)=\sum_i c_i t^{k_i}\in {\C}[t,t^{-1}]
\]
is the value of the character $f$ at $\sigma$, and $S(f(\sigma))=\sum_i c_i t^{-k_i}$.

\begin{theorem}[\cite{FJW2}]\label{theorem:4}
The characteristic map $ch: {\mathfrak R}_{\Ccrs}\lra {\mathfrak S}_{\Ccrs}$ is an isomorphism of
Hopf algebras which preserves the bilinear form.
\end{theorem}
\begin{proof}
See \cite{FJW2}. The statement corresponds to the special case of $\Gamma$ being the trivial group
there.
\end{proof}

\noindent{\bf Proof of Theorem \ref{theorem:1.5}}.  At first we see that
$\pi:\, 'H_{\Zt} \longrightarrow K_0('\mathcal H)$ is
surjective by Proposition \ref{prop:2}. We need to prove the injectivity of $\pi$.

Recall that ${\mathcal C}_n=D({\mathbf A_n\gmod})$
is the bounded derived category of finite dimensional, graded $A_n$ modules. We have the natural map
explained in Remark \ref{rmk:2.5}:
\begin{equation}\label{e:repofH}
j:\,\, '{\mathcal H}\longrightarrow {\mathfrak F}
\end{equation}
which takes ${\mathbf 1}\lba 1\rba$ to ${\mathbf 1}\lba 1\rba$,
$P_n$ to $(\fP^n,e_{(n)})$ and $Q_n$ to $(\fQ^n,e_{(n)})$. This map induces
\begin{equation*}
K_0(j):\,\, K_0('{\mathcal H})\longrightarrow K_0(\mathfrak F),
\end{equation*}
where $K_0(\mathfrak F)$ is the Grothendieck group of the category $\mathfrak F$
when only the $1$-cells and the $2$-cells are considered,
just like what $K_0('\H)$ is meant. Let
${\mathcal C}=\bigoplus_{n\ge 0}{\mathcal C}_n=\bigoplus_{n\ge 0} D({\mathbf A_n \gmod})$.
Note that $K_0({\mathcal C}_n)$ denotes the Grothendieck group of isomorphism classes
of complexes in $D({\mathbf A}_n\gmod)$, we see that an element of the image of $K_0(j)$ gives
rise to an endomorphism of $K_0(\mathcal C)\otimes_{\Z} \C$. By definition, $[M[l]]=[M]t^l$.
Thus
\begin{equation*}
K_0({\mathcal C}_n)\otimes_{\Z}\C\cong {\mathfrak R}_{n,\Ccrs} \qquad\text{ and }\qquad
K_0({\mathcal C})\otimes_{\Z}\C \cong {\mathfrak R}_{\Ccrs} \qquad \text{as $\C$ vector spaces.}
\end{equation*}
Therefore we have the sequence of ring homomorphisms:
\begin{equation*}
'H_{\Zt}\xrightarrow{\pi} K_0('\H)\xrightarrow{K_0(j)} {\rm Im}(K_0(j))
\lra{\End}({\mathfrak R}_{\Ccrs})\xrightarrow{\cong}{\End}({\mathfrak S}_{\Ccrs}),
\end{equation*}
where the last isomorphism holds due to Theorem \ref{theorem:4}. As $\{p_n,q_n\}_{n\ge 0}$
and $\{a_n\}_{n\in \Z}$ can be expressed over each other in terms of the relations
(\ref{e:pqrelation}) when the coefficients are
extended to ${\Ct}$, the composition
of the above sequence gives the Fock space representation of $'H_{\Zt}$, which is
faithful. This implies the injectivity of the composition, thus the injectivity of
$\pi$. \hfill $\square$

%
\section{Categorification of MacMahon function} \label{section:5}
%
In this section, we will describe the coefficients of MacMahon function as the dimensions of vector spaces whose elements are string diagrams.

A 3D Young diagram can be represented by a ``plane partition"\cite{OR,ORV}, i.e., a 2D array
\begin{equation*}
\pi=(\pi_{ij})^{\infty}_{i,\,j=1}=\left(\begin{array}{ccc} \pi_{11}&\pi_{12}&\cdots\\ \pi_{21}&\pi_{22}&\cdots\\ \cdots &\cdots&\cdots\\
\end{array}\right)
\end{equation*}
of nonnegative integers $\pi_{ij}\in \mathbb Z_+$, such that$ \pi_{ij}\geq\pi_{i+1,\,j},  \pi_{ij}\geq\pi_{i,\,j+1}$, where $\pi_{ij}$ is the height of the stack of cubes placed at the $(i,j)$-th position of the plane. We will identify such a plane partition with the corresponding 3D Young diagram. The total volume of the 3D Young diagram is given by
\begin{equation*}
\mid \pi\mid=\sum_{i,\,j=1}^{\infty}\pi_{ij}.
\end{equation*}

Given a plane partition $\pi=(\pi_{ij})^{\infty}_{i,j=1}$, the partition
\begin{equation*}
\pi(m)=\left\{\begin{array}{cc}(\pi_{i, \,i+m})_{i=1}^\infty, & \text{if}\  m\geq 0,\\ \\
(\pi_{j-m,\,j})_{j=1}^\infty, & \text{if}\  m< 0,\\\end{array}\right.
\end{equation*}
is the m-th diagonal slice of $\pi$. These partitions $\{\pi(m)\}_{m=-\infty}^{\infty}$ represent a sequence of 2D Young diagrams that are literally obtained by slicing the 3D Young diagrams diagonally.

The diagonal slices are not arbitrary but satisfy the condition
\begin{equation*}
\cdots\prec\pi(-2)\prec\pi(-1)\prec\pi(0)\succ\pi(1)\succ\pi(2)\succ\cdots
\end{equation*}
where $\succ$ denotes interlacing relation. We say two partitions $\lambda$ and $\mu$ of 2D Young diagram satisfy
the interlacing relation
\[\lambda=(\lambda_1,\,\lambda_2,\,\lambda_3,\cdots)\succ\mu=(\mu_1,\,\mu_2,\,\mu_3,\cdots)\]
if
\[
\lambda_1\geq\mu_1\geq\lambda_2\geq\mu_2\geq\lambda_3\geq\cdots.
\]

Let $P(N)$ denotes the number of partition $\pi$ of 3D Young diagram with $\mid\pi\mid=N$. The generating function of these numbers was studied by MacMahon\cite{B} and showed to be given by the MacMahon function:
\begin{equation*}
\sum_{N=0}^{\infty}P(N)q^N=\prod_{n=1}^{\infty}(1-q^n)^{-n}.
\end{equation*}
In statistical mechanics, this generating function becomes the partition function
\begin{equation*}
Z(q)=\sum_{N=0}^{\infty}P(N)q^N=\sum_{\pi}q^{\mid\pi\mid}=\prod_{n=1}^{\infty}\left(\frac{1}{1-q^n}\right)^n
\end{equation*}
of a canonical ensemble of plane partitions, in which each plane partition $\pi$ has an energy proportional to the volume $\mid\pi\mid$.

Consider the operators
\begin{equation*}
\Gamma_{\pm}(z)=\exp\,\left(\sum_{n>0}\frac{a_{\pm n}z^n}{n}\right)
\end{equation*}
where $a_n$ denotes the modes of the fermionic current $\phi^*\phi$
in Equation (\ref{e:fb}), therefore satisfies relation (\ref{e:1.1}).
The operators $\Gamma_\pm(z)$ can be identified with annihilation and
creation parts of the bosonic vertex operator. One can
get the following commutation relation
\begin{equation*}
\Gamma_+(z)\Gamma_-(w)=\frac{1}{1-zw}\Gamma_-(w)\Gamma_+(z).
\end{equation*}

The vertex operators $\Gamma_\pm(z)$ act on the orthonormal
bases $|\lambda \ra$ and $\la \lambda|$ as
\begin{eqnarray}
\la \lambda|\Gamma_+(z)&=&\sum_{\mu\succ\lambda}\, \la\mu|\,z^{(|\mu|-|\lambda|)},
\qquad
\Gamma_+(z) |\lambda\ra=\sum_{\mu\prec\lambda}\,z^{(|\lambda|-|\mu|)}|\mu\ra,
\label{G1}\\ \label{G2}
\Gamma_-(z)|\lambda \ra &=&\sum_{\mu\succ\lambda} z^{(|\mu|-|\lambda|)}|\mu \ra,
\qquad
\la\lambda|\Gamma_-(z)=\sum_{\mu\prec\lambda}\, \la\mu|\,z^{(|\lambda|-|\mu|)}.
\end{eqnarray}
By identifications (\ref{G1}) and (\ref{G2}) of the transfer matrices
$\Gamma_+(z)$ and $\Gamma_-(z)$, we have
\begin{equation*}
Z(q)=\la vac|\cdots\Gamma_+(q^\frac{3}{2})\Gamma_+(q^\frac{1}{2})
\Gamma_-(q^\frac{1}{2})\Gamma_-(q^\frac{3}{2})\cdots|vac\ra
\end{equation*}
where the inner product is taken in such a way that these $|\lambda\ra$ (therefore $\la\lambda|$)
for all partitions of positive integers $\lambda$ constitute an orthonormal basis
of the Hilbert space. This inner product is also given the Fock space representation of
the Heisenberg algebra $H$ and the condition that $\la vac\,|\, vac\ra=1$.

Now we replace the defining relation among $a_n, n\in\Z$ by (\ref{e:arelation}), and define
\begin{eqnarray}
\widetilde{\Gamma}_-(z)=\sum_{n=0}^\infty P_n z^n, \qquad\qquad
\widetilde{\Gamma}_+(z)=\sum_{n=0}^\infty Q_n z^n,
\end{eqnarray}
where $P_n,\,Q_n$ are the $1$-morphisms in the $2$-category $'\mathcal H$ described
in Section \ref{section:2}, and $z$ an indeterminate.

\begin{prop}\label{prop:4}
Let $\fP_n=j(P_n)$, where $j$ is the natural functor given by (\ref{e:repofH}). The following holds:
\begin{equation*}
j\left(\widetilde{\Gamma}_-(z)\right) V_\lambda=\sum_{n=0}^\infty \fP_n z^n V_\lambda=
\sum_{n=0}^\infty \sum_{\substack{\mu\succ\lambda,\\|\mu|=|\lambda|+n}} V_\mu\, z^n
\end{equation*}
where $V_{\lambda}$ denotes the irreducible representation of
the symmetric group $S_{|\lambda|}$ associated to the partition $\lambda$.
\end{prop}
\begin{proof} This result should be found in principle in \cite{FH} and
\cite{M}. As explained in Section \ref{section:3} and Section \ref{section:4},
the representation $\fP$ of $P$ induces the induction functor from ${\C}[{\mathbf S_d}] \dmod$ $\rightarrow$ $ {\C}[{\mathbf S_{d+1}}] \dmod$, which takes $V_\lambda$ to ${\rm Ind}(V_\lambda)$ for a
partition $\lambda$ of $d$, with ${\rm Ind}(V_\lambda)$ equal to the direct sum of all $V_\mu$ and $\mu$ satisfying $|\mu|=1+|\lambda|$, $\lambda \subset\mu$, all viewed in the category $\mathfrak F$
(that is, view ${\fP}:\, D({\mathbf A_d}\gmod)\rightarrow $$D({\mathbf A_{d+1}} \gmod)$ in the precise way).
The representation $j(P_n) V_{\lambda}=(\fP^n, e_{(n)})V_{\lambda}$ is isomorphic to $V_{\lambda}\otimes
V_{(n)}$, which can be decomposed as
\begin{equation*}
V_\lambda \otimes V_{(n)}=\sum_{\substack{\mu\succ\lambda,\\
|\mu|=|\lambda|+n}} V_\nu,
\end{equation*}
where $V_{(n)}$ denotes the trivial representation of $S_n$.
\end{proof}

A formula for $j\left(\widetilde{\Gamma}_+(z)\right) V_\lambda$ would be more complicated,
it is involved with the degree shifting of the complexes in $\bigoplus_{n\ge 0}D({\mathbf A_n}\gmod)$.
Instead, we consider the $2$-category $'\H$. By Equations (\ref{$1})-(\ref{$3}), we
can move all the $\widetilde{\Gamma}_+$ factors in the
infinite product $\cdots \widetilde{\Gamma}_+(q^{\frac{3}{2}})
\widetilde{\Gamma}_+(q^{\frac{1}{2}})\widetilde{\Gamma}_-(q^{\frac{1}{2}})
\widetilde{\Gamma}_-(q^{\frac{3}{2}})\cdots $ to the right hand side of
$\widetilde{\Gamma}_-$'s, and
define the $1$-morphism $\widetilde{Z}$ of the category $'\H$ as
\[
\widetilde{Z}=\la vac|\cdots \widetilde{\Gamma}_+(q^{\frac{3}{2}})
\widetilde{\Gamma}_+(q^{\frac{1}{2}})\widetilde{\Gamma}_-(q^{\frac{1}{2}})
\widetilde{\Gamma}_-(q^{\frac{3}{2}})\cdots |vac\ra.
\]
to be the complete direct summand of $\cdots \widetilde{\Gamma}_+(q^{\frac{3}{2}})
\widetilde{\Gamma}_+(q^{\frac{1}{2}})\widetilde{\Gamma}_-(q^{\frac{1}{2}})
\widetilde{\Gamma}_-(q^{\frac{3}{2}})\cdots $ which contains all the $1$-morphism
$\mathbf 1$ factors and it's degree shifting factors.
\begin{prop}\label{prop:5}
The isomorphism class $[\widetilde{Z}]\in K_0('\H)$ of $\widetilde{Z}$ is given by
\[
[\widetilde{Z}]=\prod_{n=1}^\infty\left(\frac{1}{1-q^n}\right)^n \cdot
\prod_{n=1}^\infty\left(\frac{1}{1-tq^n}\right)^n,
\]
which we denote by $Z(q,t)$.
\end{prop}
\begin{proof}
From Equations (\ref{$1})--(\ref{$3}) of Theorem \ref{theorem:1}, we get
the commutation relation
\[
[\widetilde{\Gamma}_+(z)]\cdot [\widetilde{\Gamma}_-(w)]=[\widetilde{\Gamma}_-(w)]\cdot
[\widetilde{\Gamma}_+(z)]\frac{1}{1-zw}\frac{1}{1-tzw},
\]
hence
\[
[\widetilde{\Gamma}_+(q^{k/2})]\cdot [\widetilde{\Gamma}_-(q^{l/2})]=[\widetilde{\Gamma}_-(q^{l/2})]\cdot
[\widetilde{\Gamma}_+(q^{k/2})]\frac{1}{1-q^{(k+l)/2}}\frac{1}{1-tq^{(k+l)/2}}
\]
for odd positive integers $k$ and $l$, which gives rise to the formula
for $[\widetilde{Z}]$.
\end{proof}

\begin{Remark}\label{rmk:3}
If we apply the functors $P_n$ and $Q_n$ in the categorificatin of the Heisenberg
algebra as explained in Remark \ref{rmk:2}, then the function
$[\widetilde{Z}]$ equals
\[
\prod_{n=1}^\infty\left(\frac{1}{1-t^{-1/2}q^n}\right)^n \cdot
\prod_{n=1}^\infty\left(\frac{1}{1-t^{1/2}q^n}\right)^n.
\]
\end{Remark}

The classical MacMahon function $Z(q)$ is also seen to be related to the graded vector
space of $2$-morphisms of $'\mathcal H$. Let $V_0={\C}\cdot 1$ be the 1-dimensional
vector space over $\C$. It is clear that ${\rm End}(\widetilde{\Gamma}_-(q^{\frac{1}{2}})$
$\widetilde{\Gamma}_-(q^{\frac{3}{2}})$
$\widetilde{\Gamma}_-(q^{\frac{5}{2}})\cdots)$ is a graded vector space, in
accordance with the grading of the $2$-morphisms of $'\H$.
We write
\[
{\rm End}\left(\widetilde{\Gamma}_-(q^{\frac{1}{2}})
\widetilde{\Gamma}_-(q^{\frac{3}{2}})
\widetilde{\Gamma}_-(q^{\frac{5}{2}})\cdots\right)=\mathop{\bigoplus}_{i=0}^\infty
W_i(q)
\]
where elements of $W_i(q)$ have degree $i$ and the grading depends on $q$.
We define
\[
{\mathcal Z}(q,t)={\rm gdim}\, {\rm End}\left(\widetilde{\Gamma}_-(q^{\frac{1}{2}})
\widetilde{\Gamma}_-(q^{\frac{3}{2}})
\widetilde{\Gamma}_-(q^{\frac{5}{2}})\cdots\right)=\sum_{i=0}^\infty
{\rm dim}\, W_i(q)t^i.
\]

\begin{prop}\label{prop:6}
${\mathcal Z}(q,0)=W_0(q)=Z(q)=Z(q,0).$
\end{prop}
\begin{proof}
By continuously composing with the adjoint map\,\,
$
\begin{tikzpicture}[>=stealth,baseline=-.55cm]
\draw (4.5,-.6) arc (0:190:.3) [->];
\draw (5,-.5) node{ };
\end{tikzpicture}
$, we get a degree preserving isomorphism
\[
{\rm End}(\widetilde{\Gamma}_-(q^{\frac{1}{2}})
\widetilde{\Gamma}_-(q^{\frac{3}{2}})
\widetilde{\Gamma}_-(q^{\frac{5}{2}})\cdots)\cong
{\rm Hom}( \cdots \widetilde{\Gamma}_+(q^{\frac{3}{2}})
\widetilde{\Gamma}_+(q^{\frac{1}{2}})\widetilde{\Gamma}_-(q^{\frac{1}{2}})
\widetilde{\Gamma}_-(q^{\frac{3}{2}})\cdots, \,\,{\mathbf 1}).
\]
By Equations (\ref{$1})--(\ref{$3}) of Theorem \ref{theorem:1}, we
can exchange the composition order of those transfer matrices. To get the
the $2$-morphisms of degree $0$, we need only
count the number of $\mathbf 1$ direct summands in the expansions of those
$Q_{m_1}\cdots Q_{m_s}\cdot P_{n_1}\cdots P_{n_t}$ in
\[
\cdots \widetilde{\Gamma}_+(q^{\frac{5}{2}})
\widetilde{\Gamma}_+(q^{\frac{3}{2}})
\widetilde{\Gamma}_+(q^{\frac{1}{2}})\widetilde{\Gamma}_-(q^{\frac{1}{2}})
\widetilde{\Gamma}_-(q^{\frac{3}{2}})\widetilde{\Gamma}_-(q^{\frac{5}{2}})\cdots.
\]
Therefore we get ${\mathcal Z}(q,0)=W_0(q)=Z(q).$
\end{proof}

\vspace{5pt}
When we are writing this paper we have not figured out how to calculate the whole function ${\mathcal Z}(q,t)$ or relate it to other topics, even to the function $Z(q,t)$. However we notice the $t$ parameter
has a resembling motivic/refined/deformed effect which appeared in many recent works \cite{DG, AS, BY}. One future direction is to relate our work with the so-called refined MacMahon function \cite{IKV}.
Furthermore we would like to generalize our work to give a categorified explanation of the topological
vertex which corresponds to the plane partition with three 2D Young diagram boundaries. Since Heisenberg algebra is a building block for many theories in mathematical physics, we also hope its categorification can help us reveal some facts about symmetries and invariants of those theories in the future.

\end{document}